\documentclass[pra,twocolumn,nofootinbib]{revtex4-2}

\usepackage[utf8]{inputenc}
\usepackage{amsmath}
\usepackage{amsthm}
\usepackage{enumitem}
\usepackage{amsfonts}
\usepackage{amssymb}
\usepackage{mathrsfs}
\usepackage{url}
\usepackage{graphicx}
\usepackage{comment}
\usepackage[colorlinks]{hyperref}
\hypersetup{citecolor=blue,
urlcolor=blue}
\newtheorem{theorem}{Theorem}

\newtheorem{corollary}[theorem]{Corollary}

\newtheorem{definition}[theorem]{Definition}
\newtheorem{example}[theorem]{Example}

\newtheorem{lemma}[theorem]{Lemma}

\newtheorem{proposition}[theorem]{Proposition}

\usepackage{soul}
\usepackage{color}

\definecolor{nblue}{rgb}{0.3,0.3,1.0}
\definecolor{ngreen}{rgb}{0.2,0.7,0.2}
\definecolor{nred}{rgb}{0.9,0.1,0}
\definecolor{red2}{rgb}{0.6,0.2,0.2}
\definecolor{npurple}{rgb}{0.8,0.2,0.8}
\definecolor{golden}{rgb}{0.8,0.6,0.1}
\definecolor{nsilver}{rgb}{0.3,0.4,0.5}
\definecolor{nbrown}{rgb}{0.8,0.4,0.15}
\definecolor{nrose}{rgb}{0.7,0,0.35}
\definecolor{nviol}{rgb}{0.5,0,1.0}
\definecolor{nazur}{rgb}{0,0.35,0.7}
\definecolor{nchart}{rgb}{0.2,0.4,0}
\definecolor{nbrick}{rgb}{0.55,0.25,0.15}

\newcommand{\query}{query measurement }

\def\urlprefix{}
\begin{document}
\date{\today}
\title{Local Friendliness Polytopes In Multipartite Scenarios}
\author{Marwan Haddara}
\affiliation{Centre for Quantum Dynamics, Griffith University, Gold Coast, Queensland 4222, Australia}
\author{Eric G. Cavalcanti}

\affiliation{Centre for Quantum Dynamics, Griffith University, Gold Coast, Queensland 4222, Australia}
\begin{abstract}
    Recently the Local Friendliness (LF) no-go theorem has gained a lot of attention, owing to its deep foundational implications. This no-go theorem applies to scenarios which combine Bell experiments with Wigner's friend-type set ups, containing space-like separated superobservers who are assumed to be capable of performing quantum operations on a local observer, also known as their ``friend''. Analogously to the hypothesis of local hidden variables in Bell scenarios, a set of assumptions termed ``Local Friendliness'' constrains the space of probabilistic behaviours accessible to the superobservers to be a particular subset of the no-signalling polytope in such scenarios. It has additionally been shown, that there are scenarios where the set of behaviours compatible with Local Friendliness is strictly larger than the Bell-local polytope, while in some scenarios those sets are equal. In this work, we complete the picture by identifying all the canonical Local Friendliness scenarios, with arbitrary but finite numbers of superobservers, friends, measurements and outcomes, where the set of LF correlations admits a local hidden variable model, and where they do not. Our proof is constructive in the sense that we also demonstrate how a local hidden variable model can be constructed, given a behaviour compatible with LF in the appropriate scenarios. While our principal motivation is the foundational question of better understanding the constraints from Local Friendliness, the same inequalities constraining LF polytopes have been shown to arise in a priori unrelated contexts of device-independent information processing. Our results may thus find use in those research areas as well.
\end{abstract}

\maketitle

\section{Introduction}
Research in the branch of quantum foundations often 
described as \emph{experimental metaphysics}, following terminology introduced by Shimony \cite{shimony_1993},  has recently seen a notable advance in the form of the Local Friendliness (LF) no-go theorem \cite{Bong2020}.  This result is conceptually comparable to Bell's theorem \cite{Bell1964}, in the sense that it leads to experimentally testable inequalities which, if violated, demonstrate that a certain set of seemingly reasonable metaphysical principles cannot simultaneously hold in any physical theory -- that is, in any theory that predicts the violation of those inequalities. Interestingly, it was also shown that the Local Friendliness inequalities are predicted to be violated in a fully quantum description of certain experiments, which was furthermore demonstrated in a proof-of-principle experiment \cite{Bong2020}.

The assumptions that underlie Local Friendliness are strictly weaker than those that lead to Bell's theorem \cite{Bong2020, Cavalcanti2021} and hence the result poses new challenges in quantum foundations. A number of works have already emerged \cite{Cavalcanti2021, Cavalcanti2021Foundphys, DiBiagio2021, Haddara_2023, Wiseman2023thoughtfullocal, utrerasalarcón2023, yīng2023relating, schmid2023review, Yang2022, Ding2023, Baumann2023} that expand upon, explore, or discuss the significance of the Local Friendliness result from both theoretical and experimental perspectives.  Given this considerable interest, we believe a  more thorough investigation of the mathematical constraints implied by Local Friendliness is in order.

In the original work of \cite{Bong2020} the no-go result was demonstrated in an extension of the original Wigner's friend \cite{Wigner1961} scenario, by describing a protocol in which two technologically advanced space-like separated agents are allowed to perform arbitrary operations on their``friend''. It was shown that correlations compatible with  Local Friendliness form a convex polytope in such scenarios and hence are bounded by a finite set of linear inequalities. Depending on the number of possible measurements the parties are able to subject their friend to, the Local Friendliness polytope was furthermore shown to be either equal to the Bell polytope -- which bounds the correlations compatible with a local hidden variable (LHV) model -- or to be a strict superset of the Bell polytope. Similar correlation sets had also been studied independently in \cite{Woodhead2014} under the name of ``partially deterministic polytopes'', motivated by tasks in device-independent information processing. 
  
In this work, we generalize the scenarios of \cite{Bong2020} to include arbitrary number of parties, measurements and outcomes per party and consider cases where arbitrary subsets of the distant laboratories contain a friend. We explore the mathematical properties of correlations compatible with the assumption of Local Friendliness in such scenarios, extending the present knowledge beyond the bipartite scenarios of \cite{Bong2020} and \cite{Woodhead2014} (see also \cite{utrerasalarcón2023}). In particular we identify all those scenarios where the Local Friendliness polytope equals the Bell polytope, hence identifying when ``genuine'' Local Friendliness inequalities (i.e. facets of the LF polytope that are not facets of the corresponding Bell polytope) may arise. Our proofs are constructive in the sense that whenever the Local Friendliness set equals the LHV set, we also show how the LHV model can be constructed given a behaviour compatible with Local Friendliness.

This work is structured as follows. In Section \ref{CanonicalScenariosSection} we describe the type of scenarios considered in this work, and recall the definition of Local Friendliness. In Section \ref{basicDefinitionsSection} we review the basic definitions, mathematical tools and previously known results related to this problem. We also present some basic facts about sets compatible with Local Friendliness. Our main results are contained in Section \ref{ResultsSection}, where we provide an analytical classification of Local Friendliness polytopes. This classification is complete in the sense that one can identify not only the scenarios where the LF correlations always admit an LHV model, but also all the distinct scenarios where the correlations compatible with Local Friendliness are equal to each other. In Section \ref{DiscussionSection} we demonstrate our results in a few example scenarios of particular interest, and discuss their significance in a broader context. We finish with concluding remarks in Section \ref{ConclusionSection}.

\section{Canonical Local Friendliness scenarios \label{CanonicalScenariosSection}} The original Local Friendliness no-go result \cite{Bong2020} was exhibited in a type of extended Wigner's friend scenario suggested by Brukner \cite{Brukner2017,Brukner2018}, which combines ingredients from the Bell \cite{Bell1964} and Wigner's friend \cite{Wigner1961} set ups. We will henceforth refer to such scenarios as `canonical' or `standard' LF scenarios, since recent works have extended the investigation of the implications of Local Friendliness to more exotic situations which include sequential measurements \cite{utrerasalarcón2023} or that dispense with space-like separation \cite{yīng2023relating}. It is conceivable that one could devise further protocols where the assumptions underlying Local Friendliness have nontrivial consequences. Discussion of nonstandard scenarios is, however, beyond the scope of this work. 

A general canonical LF scenario $S$ consists of a number N of space-like separated superobservers or agents\footnote{Here we will refer to the superobservers but not the observers as ``agents'', because in the scenarios we consider only the superobservers have choices of measurements. Note however that the observers may also be agents in other senses, e.g.~they may be AI agents running in a large quantum computer~\cite{Wiseman2023thoughtfullocal}.}, which we index  by a set $I_A = \{ 1 \ldots N \}$.  Each superobserver $i$ has a number of measurements labelled by $x_i \in M_i = \{ 1 \ldots m_i \}$ to choose from, and for each measurement a number of possible outcomes $a_{x_i} \in O_{x_i} = \{ 1\ldots o_{x_i} \}.$ Some of the superobservers have a ``friend'', which we will also refer to as simply an `observer'. The index set of the friends will be taken to be a subset $I_F \subset I_A$ with the reading that if $i \in I_F$, then the $i$-th superobserver has a friend. We use the convention where $\subset$ denotes the basic  `is a subset of' relation, while $\subsetneq$ is used if emphasis on the subset being strict is warranted. Let us define $O = \{O_{\Vec{x}} \}_{\Vec{x}\in M}$ so a scenario $S$ may be specified by the quadruple $S = (I_A, I_F, M, O)$.  To avoid complications due to trivial sets of parameters, we assume that the sets $I_A, M_i, O_{x_i}$  have cardinality at least 2 for all $i, x_i$. For an arbitrary set $V$, we denote its cardinality by $|V|$. Lastly, it will also be useful to introduce $S_p= (I_A, M, O)$, the triple of parameters which define the \emph{public scenario}. We will occasionally denote $S = (S_p, I_F )$ for a canonical LF scenario $S$, where found useful. Conceptually we  may identify $S_p$ with the standard Local Friendliness scenario in which $I_F = \emptyset$.

In the experimental protocol, a common source sends systems to each site $i$, which is then measured by an observer if one is located at that site. Each observer will record some outcome, denoted as $c_i$ for a friend at site $i$, and store it in their memory. From the perspective of the superobservers, the friends are physical systems, on which by assumption (close to) arbitrary quantum operations can be implemented, hence the reason for the descriptive prefix `super' in `superobserver'. 

After the observers have recorded their outcomes, the entire contents of their laboratory will be subject to some measurement by the superobserver at their site. One measurement, which by convention is chosen to be $x_i =1$, will be given a special role: it corresponds to the case where the agent reads their friend's memory and assigns $a_i =c_i$, so that their outcome matches that of the friend\footnote{This requires the complementary assumption that the reading of the friend's memory corresponds to the outcome observed by the friend, which was called ``Tracking'' in \cite{schmid2023review}.}. This input shall be referred to as the ``query measurement". If a friend does not exist at site $i$, then no special constraints are imposed on any of the measurements. The case where $I_F = \emptyset $ corresponds to an ordinary Bell scenario. An example of a canonical LF scenario is portrayed in Fig.~\ref{fig:ThreepartyScenario}.  \begin{figure}
    \centering
    \includegraphics[width=\columnwidth]{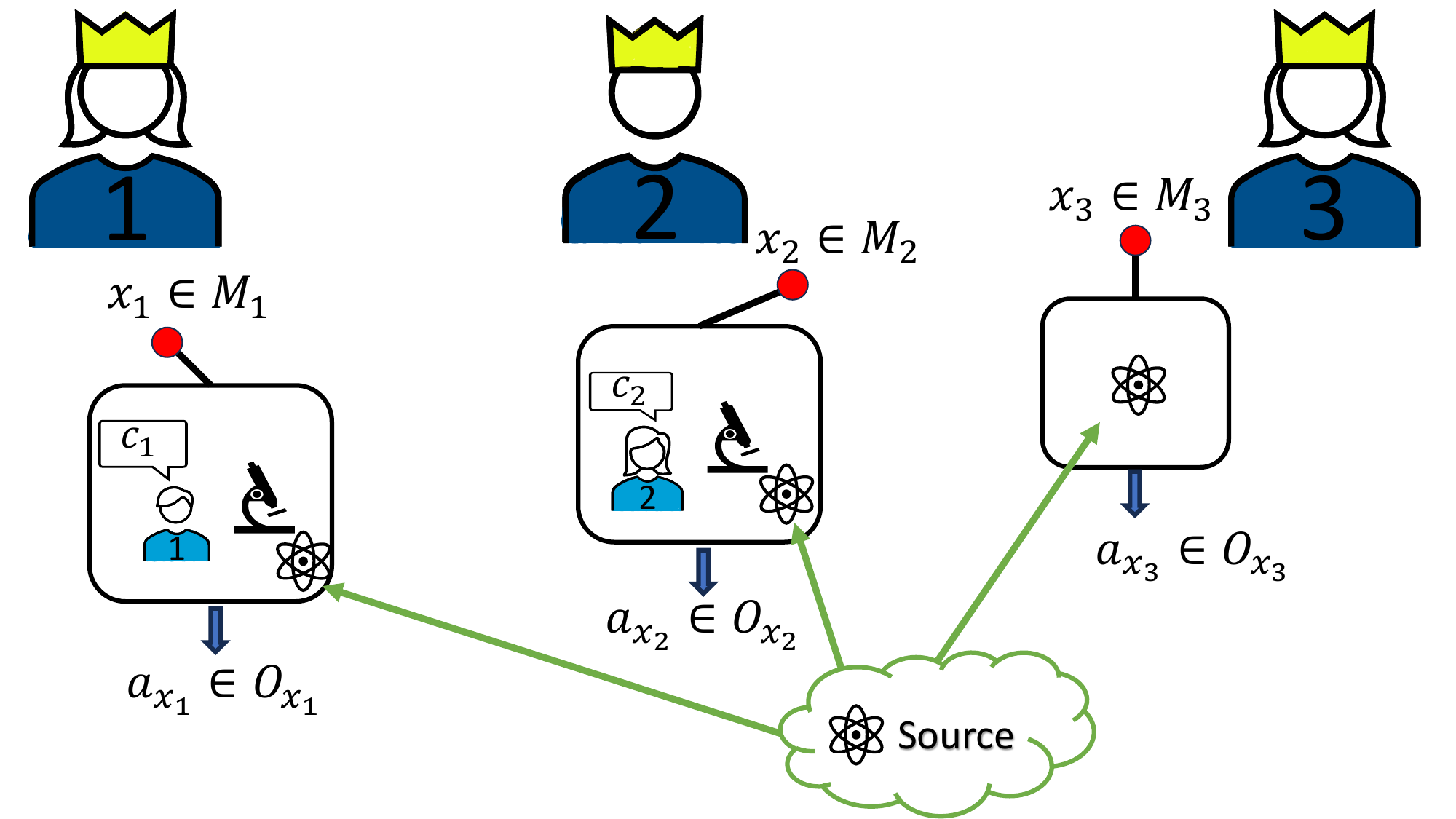}
    \caption{An illustration of a canonical Local Friendliness scenario with $I_A = \{1, 2, 3\}$ and $ I_F = \{1, 2 \}$. A source sends systems to three distant laboratories. The superobservers are depicted wearing crowns and shirts indicating their label. The superobservers have sets $M_1, M_2$ and $M_3$ of interventions they can perform on the contents of the boxes. Following the interventions,  their  devices produce some outcomes $a_{ x_i}$ from sets $O_{x_i}$. Superobservers 1 and 2 have a friend, who observe a physical system and record outcomes $c_i$ inside their sealed laboratories.  According to the protocol for the inputs inputs $x_i = 1$  of the first two superobservers it must be that $a_{x_i = 1} = c_i$ while no such constraint is imposed on the corresponding input $x_3 =1$ of superobserver 3.  }
    \label{fig:ThreepartyScenario}
\end{figure}

In line with the black-box picture of Bell's theorem, the relevant mathematical objects in these types of scenarios are the probability distributions

\begin{align}
    \wp(a_{x_1}, a_{x_2}, \ldots,  a_{x_N}| x_1 , x_2, \ldots , x_N ),\label{primitiveBehaviour}
\end{align}
that is, a distribution for each set of measurement choices. These distributions represent the statistics collected by the agents at the end of the experiment, and their collection is in the context of Bell experiments commonly referred to as the \emph{behaviour}~\cite{Brunner2014}. 

Given the generality of our presentation, notation of the kind used in Eq.~(\ref{primitiveBehaviour}) quickly becomes cumbersome. We therefore introduce the string notation where the expression of Eq.~(\ref{primitiveBehaviour}) is replaced with  $\wp(\Vec{a}| \Vec{x})$, where $\Vec{x}  \in M =  M_1 \times M_2 \times \ldots \times M_N$ is the string representing the choices of measurements and $\Vec{a} \in O_{\Vec{x}}= O_{ x_1}\times O_{x_2} \times \ldots \times O_{x_N}$ is a string of outcomes for those measurement choices. Note that we have dropped the index $\Vec{x}$ from the subscript of $\vec{a}$. This should lead to no confusion, since the measurement context can be taken to be  implicit in the notation of the conditional distributions.  For convenience, we may therefore  simply denote the outcome of the  $i'$th superobserver in the string $\Vec{a}$ by $a_i$. If the detail is useful or significant, we will occasionally revert this back and indicate $a_{x_i}$ to represent the outcome of agent $i$ for their intervention $x_i$, or $a_{x_i = s}$ if emphasis for the chosen measurement $s$ is warranted.

Since we will also be considering marginals of behaviours and cases where the distributions are independent of some of the inputs $x_i$, we may also denote 
\begin{align}
    \sum_{a_i, i \in E} \wp (\Vec{a}|\Vec{x}) := \wp(\Vec{a}_J|\Vec{x}), \hspace{0.4cm}J = (I_A \setminus E), \label{substringNotationEquation}
\end{align} 
with $\Vec{a}_J $ referring to a substring of $\vec{a}$ consisting of the elements indexed by $j \in J$. When referring to some individual elements such as $\vec{x}_j$, we will sometimes omit the arrow and take it to be implicitly understood that by $x_j$ we refer to the $j$'th element of the string $\vec{x}$.

We note that despite the inclusion of the friends in the picture, the behaviour is still the relevant object of interest since as per the description of the protocol no trace of the friends outcomes have to exist at the end of the experiment for all measurements. It is this simple observation, that can be used to justify the intuition why the assumption of Local Friendliness has nontrivial implications in scenarios where $I_F \neq \emptyset$. 

\definition{(Local Friendliness)\\
\emph{Local Friendliness} is the conjunction of the following two principles\label{LocalFriendliness}:
\begin{enumerate}[label=\roman*), ref=\ref{LocalFriendliness}.\roman*]
    \item  (Absoluteness of Observed Events\label{axiomAOE}) (AOE)
Every observed event is an absolute single event, not relative to anything or anyone.

    \item  (Local Agency) \label{LA} (LA)
The only relevant events correlated with an intervention are in its future light cone. 
\end{enumerate}
}
\normalfont

  In particular, the assumption of Absoluteness of Observed Events in Def.~\ref{axiomAOE} justifies the assignment of a probability distribution over the  observed outcomes of all observers, including in cases where the superobservers do not read those records. In the canonical LF scenario, it thus implies the existence of a distribution $
    P(\vec{a}\vec{c}|\vec{x}) $ such that the behaviour can be recovered by marginalization:
\begin{equation}
    \wp(\vec{a}|\vec{x}) = \sum_{\vec{c}} P(\vec{a}\vec{c}|\vec{x}). \label{JointforLF}
\end{equation} 

Additional constraints arise from Local Agency, which implies that
\begin{equation}
     P(\vec{a}\vec{c}|\vec{x}) = P(\vec{a}| \vec{x}, \vec{c}) \cdot P(\vec{c}), \label{LFLocalAgency1}
\end{equation}
and that for any subset of superobservers $V \subset I_A$ and $\forall \Vec{a}_V, \Vec{x}$
\begin{equation}
    P(\Vec{a}_V| \Vec{x}, \Vec{c}) = P(\Vec{a}_V|  \Vec{x}_V, \Vec{c}) \label{LFlocalAgency2}
\end{equation}
owing to the space-time relations of the experiment. 

Under the assumption of Absoluteness of Observed Events, the protocol for the \query where a superobserver asks their friend for their outcome  justifies the mathematical condition that for all subsets $F_{\Vec{x}}$ of parties that perform that measurement in a given run, their observed outcome $a_i$ will be equal to the outcome $c_i$ observed by their friend. That is, $\forall i\in \{ i \in I_F| \vec{x}_i =1   \} = F_{\Vec{x}} \subset I_F$
\begin{align}
    P(\vec{a}_{F_{\Vec{x}}} |\vec{x}, \vec{c}) = \prod_{i \in F_{\vec{x}}}\delta_{a_i,c_i}. \label{LFspecialMeasurement}
\end{align}

We emphasize that the condition in Eq.~\eqref{LFspecialMeasurement} is context sensitive in the sense that  the decomposition holds only for those parties $i$ which have $x_i =1$ given a string $\Vec{x}$. Notation such as $F_{\vec{x}}$ will be used throughout the paper for  context-dependent sets. 

  \begin{figure}
    \centering
    \includegraphics[width=\columnwidth]{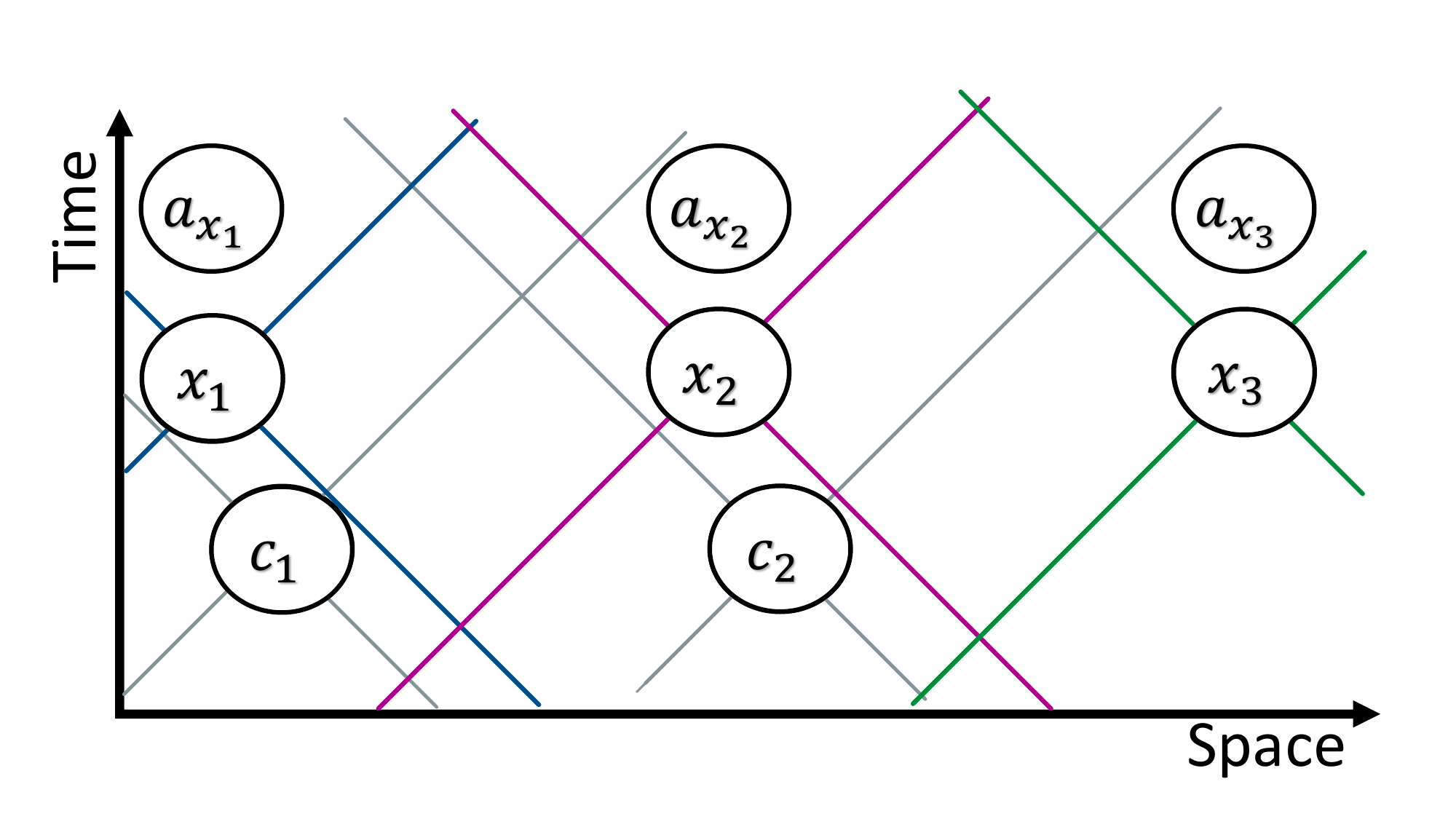}
    \caption{\label{fig:spacetimeThreeparty}
    The space-time diagram corresponding to the scenario of Fig.~\ref{fig:ThreepartyScenario}. The only variable in the future light cone of $x_i, i \in \{1,2,3 \}$, is $a_i$, and hence by Local Agency of Def.~\ref{LA} that intervention is uncorrelated with everything else. 
    }
\end{figure}

  Figure \ref{fig:spacetimeThreeparty} illustrates the space-time relations of the scenario depicted in Fig.~\ref{fig:ThreepartyScenario}. As an illustration, we demonstrate the implications of Def.~\ref{LocalFriendliness} in this scenario. 
  
  In the situation of Fig.~\ref{fig:ThreepartyScenario}, Eqs.~\eqref{JointforLF} - \eqref{LFLocalAgency1} can be used to decompose any behaviour as 
  \begin{align}
      \wp(\Vec{a}|\Vec{x}) = \sum_{c_1,c_2} P(a_1,a_2, a_3|x_1,x_2,x_3 ,c_1, c_2)P(c_1, c_2). \label{ExampleFigureEquation}
  \end{align}  Equation \eqref{LFlocalAgency2} further implies that these distributions satisfy the marginal constraints 
  \begin{align}
      P(a_i,a_j | x_1, x_2, x_3, c_1, c_2) = P(a_i,a_j | x_i, x_j, c_1, c_2)
  \end{align} for all $(i\neq j), i,j \in \{1,2,3 \}$ and \begin{align}
      P(a_i|x_1,x_2,x_3,c_1 c_2) = P(a_i|x_i,c_1, c_2) 
  \end{align} for all $i \in \{ 1,2,3 \}$. These  identities can be used in conjunction with Eq.~\eqref{LFspecialMeasurement} and the definition of conditional probability to get 
  \begin{align}
      &P(a_1,a_2,a_3| x_1=1,x_2=1,x_3, c_1,c_2) \\
      &= P(a_3|x_3,c_1,c_2)P(a_1, a_2|a_3, x_1=1,x_2=1,x_3 ,c_1,x_3) \\
      &= P(a_3|x_3,c_1,c_2)\delta_{a_1,c_1} \delta_{a_2,c_2}
  \end{align}  for arbitrary $x_3$ and 
  \begin{align}
      &P(a_1,a_2,a_3| x_1,x_2,x_3, c_1,c_2) \\
      &= P(a_3, a_j|x_j, x_3, c_1, c_2)P(a_i|a_j, a_3, x_i, x_j, x_3, c_1, c_2) \\
     & = P(a_3, a_j|x_j, x_3, c_1, c_2) \delta_{a_i, c_i}
  \end{align} for $i,j \in \{1,2 \}, i\neq j$, $x_j \neq x_i = 1$. Gathering these identities, Eq.~\eqref{ExampleFigureEquation} may be rewritten as 
  \begin{align}
      \wp(\Vec{a}|\vec{x}) = \begin{cases}
      \sum_{\Vec{c}}  \delta_{a_1, c_1} \delta_{a_2,c_2} P(a_3|x_3, \Vec{c})P(\Vec{c}) & x_1, x_2 = 1 \\ 
      \sum_{\Vec{c}} \delta_{a_i, c_i} P(a_j, a_3| x_j, x_3, \Vec{c})P(\Vec{c}) & x_j \neq x_i =1, \\
          \sum_{\Vec{c}} P(\Vec{a}|\Vec{x})P(\Vec{c})  & x_1, x_2 \neq 1.\\
      \end{cases}
  \end{align}
  Emphatically, a priori,  the decomposition of $\wp(\Vec{a}|\Vec{x})$ implied by Def.~\ref{LocalFriendliness} therefore depends on the inputs $\Vec{x}$ in general.

Aside from this brief motivation, we do not intend to discuss the significance of these metaphysical principles further. In this work, we are mainly interested in the mathematical characterization of behaviours compatible with Def.~\ref{LocalFriendliness} in different scenarios. We are now ready to present the relevant mathematical definitions and tools.

\section{Basic definitions and mathematical tools \label{basicDefinitionsSection}}

Given the conventions of Section \ref{CanonicalScenariosSection}, we proceed with the relevant definitions.

\definition{(LHV behaviours) \label{LHVbehaviours}\\
A behaviour is called \emph{LHV-modelable}, or \emph{Bell-local}, if and only if it can be given a local hidden variable (LHV) model. That is, iff
\begin{align}
    \wp(\vec{a}| \vec{x}) = \int_{\Lambda} d\lambda \, p(\lambda) \left[\prod_{i=1}^N P(a_i|x_i, \lambda)\right] \label{Bell-localmaths},
\end{align}
with $a_i \in \Vec{a}, x_i \in \Vec{x}, $ for some set $\Lambda$ and measure $p(\lambda) \geq 0, \int_{\Lambda} p(\lambda) = 1$.
The set of LHV behaviours in a given scenario $S$ is denoted by $\mathcal{B}(S)$.
}

\normalfont

Since the observers do not enter Def.~\ref{LHVbehaviours} it is clear that $\mathcal{B}(S_p, I_F) = \mathcal{B}(S_p)$.

We recall the following well-known result on LHV models often attributed to Arthur Fine \cite{Fine1982}. 

\theorem{(Fine's theorem\footnote{The relevant elements of the theorem presented here can be found in Fine's work \cite{Fine1982}, though the theorem there considers only a specific bipartite scenario and contains other statements as well. For alternative presentations of this result, or parts of it,  we refer the reader for example to \cite{Werner2001}, \cite[Sec. II.B]{Brunner2014} or \cite[Sec. 8.]{Abramsky2011}.}) \label{FinesTHRM} \\
The following three statements are equivalent. 
\begin{enumerate}[label=\roman*), ref=\ref{FinesTHRM}.\roman*]
    \item The behaviour $\wp(\vec{a}|\Vec{x})$ is LHV modelable in the sense of Def.~\ref{LHVbehaviours}. \label{FinesLHVmodelable}

    \item There exists a joint distribution \label{FinesJointDistribution} ${P}(\vec{\alpha})$, $\vec{\alpha} = (\alpha_{x_1 = 1},\ldots ,\alpha_{x_1 = m_1}, \ldots ,\alpha_{x_N = m_N}) \in O_{x_1 = 1} \times \ldots \times O_{x_1 = m_1} \times \ldots \times O_{x_N = m_N} $
    over all the potential outcomes of all possible measurements for all agents, which recovers all the distributions in the behaviour $\wp(\vec{a}|\Vec{x})$ by appropriate marginalization. That is,
    \begin{equation}\label{JointdistributionMarginalizationEq}
     \wp(\vec{a}_{\vec{x}}|\vec{x}) = P(\vec{\alpha}_{\vec{x}} = \vec{a}_{\vec{x}}) 
    \end{equation}
    where 
    \begin{align}
          P(\vec{\alpha}_{\vec{x}}) := \sum_{\alpha_{x_i'}: x_i' \neq x_i} {P}( \Vec{\alpha}).
    \end{align}

    \item \label{FinesDeterministicDistribution} The behaviour $\wp(\vec{a}|\Vec{x})$ admits a fully deterministic separable model of the form
    \begin{equation}\label{FinesDeterministicEquation}
        \wp(\Vec{a}|\Vec{x}) = \sum_\lambda P(\lambda)\, \left[ \prod_{i}\delta_{a_i, \alpha_{x_i}(\lambda)}\right].
    \end{equation} Here  $\alpha_{x_i}: \Lambda \rightarrow O_{x_i}$ while $a_i$ and $x_i $ are elements of the strings $\vec{a}, \vec{x}$ for all $i$ respectively.
\end{enumerate}
}

\begin{proof}
    Various proofs of parts of this result exist in the literature. We provide a proof for completeness.

    Clearly the case \ref{FinesDeterministicDistribution} is a special case of \ref{FinesLHVmodelable} so the implication $\ref{FinesDeterministicDistribution} \Rightarrow \ref{FinesLHVmodelable}$ holds. To show that these are all equivalent, it is therefore sufficient to demonstrate that \ref{FinesLHVmodelable} $\Rightarrow\ref{FinesJointDistribution} \Rightarrow \ref{FinesDeterministicDistribution}$. Let us then assume that \ref{FinesLHVmodelable} holds. Then, in particular, every distribution of the form $P(a_{x_i}| x_i, \lambda)$ is well defined, so one can construct a joint by taking the product as 
    \begin{equation}
        {P}(\Vec{\alpha}) = \int_{\Lambda} d\lambda p(\lambda)  \left[  \prod_{(i,s) \in I_A \times M_i} P(a_{x_i=s}|x_i, \lambda) \right].
    \end{equation}
    By definition, this recovers the probabilities of the behaviour $\wp(\vec{a}|\Vec{x})$ compatible with assumption \ref{FinesLHVmodelable} as marginals, so \ref{FinesLHVmodelable} $\Rightarrow \ref{FinesJointDistribution}$ holds. Suppose then that \ref{FinesJointDistribution} holds. Since then, 
    \begin{equation}
        \wp(\Vec{a}|\Vec{x}) =  \sum_{\alpha_{x_i'}: x_i' \neq x_i} {P}( \Vec{\alpha}) =  \sum_{\Vec{\alpha}} \left(  \prod_{i \in I_A} \delta_{a_i, \alpha_{x_i}(\vec{\alpha})}  \cdot {P}( \Vec{\alpha})\right) ,
    \end{equation} it is evident that the  implication \ref{FinesJointDistribution} $\Rightarrow \ref{FinesDeterministicDistribution}$ holds as well, by identifying $\vec{\alpha}$ with $\lambda$ of Eq. \eqref{FinesDeterministicEquation}. That is, the 
   distribution ${P}(\Vec{\alpha})$ can now be understood as the distribution over a finite number of LHVs  which determine the outcomes for all the measurements. 
\end{proof}
\normalfont

Although Fine's theorem is phrased here with reference to operational quantities such as the behaviour $\wp(\vec{a}|\vec{x})$, one may consider it as a purely mathematical result which relates constraints, as phrased in \ref{FinesLHVmodelable} and \ref{FinesDeterministicDistribution}, on a family of probability distributions (here the behaviour) to the existence of a single probability space in which all the variables can be embedded while recovering the same family of distributions as marginals, as in \ref{FinesJointDistribution}. Phrased this way, Theorem \ref{FinesTHRM} can be understood to be related to the general kinds of probability extension problems considered, for example, by Vorob'ev \cite{Vorobeev1962}. 

\definition{(Quantum behaviours\label{quantumbehaviour})\\
The set $\mathcal{Q}(S_p)$ of quantum behaviours in a public scenario with $N=|I_A|$ agents is the set which can be modelled by a suitable
Hilbert space of the form $\mathcal{H}_{\Vec{A}} = \mathcal{H}_{A_1} \otimes \mathcal{H}_{A_2} \otimes \ldots \otimes \mathcal{H}_{A_N}$, a quantum state $\rho$, i.e. a positive unit-trace linear operator $\rho: \mathcal{H}_{\Vec{A}} \rightarrow \mathcal{H}_{\Vec{A}} $, and a set of POVM's with elements $M_{a_i|x_i}$, i.e.~for which $M_{a_i|x_i} \geq 0 \; \forall i,x_i, a_i$ and $\sum_{a_i \in O_{x_i}} M_{a_i|x_i} = I_{\mathcal{H}_{A_i}} \forall i, x_i$, such that
\begin{align}
    \wp(\Vec{a}|\Vec{x}) = \mathrm{tr}[M_{a_1|x_1} \otimes M_{a_2|x_2} \otimes \ldots \otimes M_{a_n | x_N} \rho ]. \label{QuantumGeneralBornRule}
\end{align}
}
\normalfont
Definition \ref{quantumbehaviour} is well-known from the context of Bell scenarios and may be understood as depicting a situation where $N$ agents each perform a local measurement $x_i$ on their share of the joint quantum state. While the portrayal of the scenario in Section \ref{CanonicalScenariosSection} was emphatically theory-independent, part of the appeal of the Local Friendliness no-go result is the fact that quantum descriptions of such scenarios allow for demonstrating the violation of its assumptions \cite{Bong2020,Haddara_2023}. Since the inclusion of the friends in the picture essentially means adding more physical systems to the scenario,  one should not, to be precise, postulate outright that the set of quantum correlations $\mathcal{Q}(S_p)$ defined for the Bell case equals that of $\mathcal{Q}(S)$ of a general canonical LF scenario since the additional physical systems and steps in the protocol ought to be included in the model. More attentively, it is clear that $\mathcal{Q}(S) \subset \mathcal{Q}(S_p)$, as the agents are treating their friends as parts of the experiment, but the converse inclusion is not self-evident from the description of the protocol. We will therefore provide an argument of why Def.~\ref{quantumbehaviour} is adequate in the general case as well.

Let us first recall that by application of Naimark's theorem, see for example \cite{Busch2016}, any behaviour representable in the form of Eq.  \eqref{QuantumGeneralBornRule} can from a mathematical perspective be equivalently be expressed by some sets of $N$ local projective measurements $\hat{P}_{a_i|x_i}$ on a pure state $\rho \equiv \hat{P}_{\eta}, \eta \in \mathcal{H}_{\Vec{S}}$, by enlarging the Hilbert spaces $\mathcal{H}_{\mathcal{A}_i}$ to $\mathcal{H}_{A_i} \otimes \mathcal{H}_{e_i}$, if necessary. Physically this may be understood as the agents taking a part of the environment surrounding the systems into account as well. Since we may assume that in a general LF scenario the friends are contained in a fully sealed laboratory over which the agents have complete control, we are (arguably) justified to model the entire experiment by pure states and projective measurements. For our purpose the question whether this type of argument is watertight from a physical perspective is in any case in fact secondary, since our claim essentially is that  a quantum model with the right reading exists for any canonical LF scenario.

\proposition{Any quantum behaviour $\wp(\Vec{a}|\Vec{x}) \in \mathcal{Q}(S_p)$ can be realized in a standard LF scenario with $I_F$ arbitrary, consistently with the description of the protocol, assuming the evolution due to measurements can be described unitarily. \label{QuantumBehaviourIndependentofIC}}
\begin{proof}
    Suppose $\wp(\Vec{a}|\Vec{x})$ is a behaviour which may be reproduced in a Bell scenario by projective measurements $\hat{P}_{a_{x_i}|x_i}$ on a pure state $\hat{P}_{\eta}$,  $\eta \in \mathcal{H}_{\Vec{S}} = \bigotimes_{i=1}^N \mathcal{H}_{S_i}$. For each $i \in I_F \cap I_A = I_F$ introduce a Hilbert space $H_{F_i}$ and a pure state $\hat{P}_{R_i}$ representing the contents of the entire laboratory of the friends aside from the system $S_i$, with the index $R_i$ representing the 'ready' state of the laboratory in the beginning of the experiment. For consistency we require that $\mathrm{dim} (\mathcal{H}_{F_{i}}) \geq |O_{x_i = 1}|$, to guarantee that the distinct outcomes of the measurements can be faithfully read from the friends' records. 

    The interaction between the friend and the system may be modelled by unitaries $U_{i}: \mathcal{H}_{F_i} \otimes \mathcal{H}_{S_i} \rightarrow \mathcal{H}_{F_i}  \otimes \mathcal{H}_{S_i},$ which correlate the system states with their laboratory, so that  
    \begin{equation}
        \prod_{i \in I_F } \hat{P}_{R_i}^{\otimes} \otimes \hat{P}_{\eta} \mapsto \hat{P}_{\Vec{U} \eta} = U \hat{P}_{\Vec{\eta}} U^\dagger.
    \end{equation}
The requirement for the unitary is that for every  $i \in I_F $
\begin{align} \label{subspaceConditionQuantum}
    & \mathrm{supp}[ U_i (\hat{P}_{R_i} \otimes \hat{P}_{a_i|x_i =1} ) U_i^\dagger ] \subset \\
   & \mathrm{supp} [(\hat{R}_{a_i| x_i =1})] \otimes  \mathrm{supp} [(\hat{P}_{a_i|x_i =1})], 
\end{align}
     where the support of an operator $A$ is defined as the set $\mathrm{supp} [A] = \{x| A(x) \neq 0 \}$. Here  $\hat{R}_{a_i|x_i =1}$ is an orthogonal projector on $\mathcal{H}_{F_i}$ for each distinct $a_i$ with the property that $\sum_{a_i} \hat{R}_{a_i| x_i =1} = I_{\mathcal{H}_{F_i}}. $ First of all, unitaries which such conditions exist, as one may understand it as a form of Wootters-Zurek cloning \cite{Wootters1982} where the basis elements of $\mathcal{H}_{S_i}$ in the support of $\hat{P}_{a_i|x_i =1}$ are copied into basis elements in the support of $\hat{R}_{a_i| x_i =1}$. The normalization of the $\hat{R}_{a_i| x_i =1}$ can be guaranteed by adding the remaining basis elements of $\mathcal{H}_{F_i}$ arbitrarily to the range of some $\hat{R}_{a_i| x_i =1}$ i.e by defining $R_{a_i|x_i} = \mathcal{I}_{\mathcal{H}_{F_i}} -\sum_{j \neq i} R_{a_j|x_j}$ for some $i$. Second  the $\hat{R}_{a_i| x_i =1} \otimes \hat{P}_{a_i|x_i =1} = \Tilde{M}_{a_i | x_i =1} $ are by construction PVM's and  by defining $\Tilde{M}_{a_i |x_i} = U_i\hat{P}_{a_i|x_i = 1}U_i^*$  for any $i \in I_F$ with $x_i \neq 1$  one sees that  
\begin{align}\label{traceRecoveredSpecialsetting}
        \mathrm{tr}[ (\bigotimes_{i\in I_A} \Tilde{M}_{a_i |x_i}) \hat{P}_{\vec{U} \eta}] = \mathrm{tr}[(\bigotimes_{i\in I_A}  \hat{P}_{a_i|x_i})  \hat{P}_{\eta} ],
     \end{align}
where $\Tilde{M}_{a_i | x_i} = \hat{P}_{a_i|x_i}$ whenever $i \notin I_F$. Eq.~\eqref{traceRecoveredSpecialsetting} means that the probabilities of the arbitrary Bell-experiment can be lifted to the scenario with nonempty set of friends while holding on to the interpretation of the condition for the \query $x_i =1$ corresponding to reading the record of the outcome by the observer.  
\end{proof}
\normalfont
The proof of Prop.~\ref{QuantumBehaviourIndependentofIC} shows that at least if the friends' measurements are modelled unitarily, then the superobservers can lift arbitrary correlations producible from the $N$ quantum systems sent by the source in the initial stage of the protocol. It may be the case that this does not hold for alternative theories describing the evolution of the states of the friends, for example if their observations cause an objective collapse of the quantum state. While we will later establish a few general results that hold for the set $\mathcal{Q}(S) = \mathcal{Q}(S_p)$ relative to Def.~\ref{quantumbehaviour}, exploring this issue further is not the main purpose of this work, and we shall take Def.~\ref{quantumbehaviour} as valid in arbitrary canonical scenarios. 

  \definition{(No-signalling behaviours\label{No-signallingbehaviour})\\
  The set $\mathcal{NS}(S)$ is the set of behaviours that obey the no-signalling constraints, that is $\forall (i, x_i \neq x_i')$
  \begin{align}
    \sum_{a_i} \wp(\Vec{a}| x_1\ldots x_i \ldots x_N) = \sum_{a_i} \wp( \vec{a}|x_1 \ldots x_i' \ldots x_N). \label{no-signalling}
 \end{align}}

\normalfont
Similarly to the case of the LHV set in Def.~\ref{LHVbehaviours}, it is clear that $\mathcal{NS}(S)=\mathcal{NS}(S_p)$. 

We will mostly use the notation  $K(S)$  for the sets $K \in \{\mathcal{B}, Q, \mathcal{NS}\}$ in arbitrary canonical LF scenarios $S$, despite the established identities $\mathcal{K}(S)=\mathcal{K}(S_p)$ as strictly speaking a scenario $S$ has physically identifiable structure not present in $S_p$. We will occasionally revert back to denoting by $K(S_p)$  any one of those sets if emphasis on the fact that the mathematical features of those sets do not depend on $I_F$ is found useful or convenient.  

A comprehensive characterization of the sets $\mathcal{B}(S), \mathcal{Q}(S)$ and $\mathcal{NS}(S)$ can be given by a geometric approach, see for example \cite{Brunner2014} for a review. A way to formalise this perspective is to note that the behaviour $\wp(\Vec{a}|\Vec{x})$ may be represented as points in $\mathbb{R}^D$ with $D = \sum_{\vec{x}\in M}  |O_{\vec{x}}|$ the number of probabilities for the outcomes in all contexts in the behaviour.  While each  behaviour satisfies $\forall (\Vec{a}, \vec{x}) $ $\wp(\Vec{a}|\Vec{x}) \geq 0$ and $\forall \vec{x}$ $\sum_{\Vec{a}}\wp(\Vec{a}|\Vec{x}) = 1 $, the additional constraints in Defs.~\ref{LHVbehaviours}, \ref{quantumbehaviour} and \ref{No-signallingbehaviour} further bound each of the sets respectively. 

By theorem \ref{FinesDeterministicDistribution} the set $\mathcal{B}(S)$ can be completely classified as the convex hull of a finite set of points, which is the defining feature of an object known as a convex polytope \cite{Grünbaum2003}. Such a polytope can be alternatively expressed as the intersection of a finite set of half spaces, which is to say that a finite number of inequalities of the type $\Vec{\wp} \cdot \vec{r} \leq R$, with $\vec{r}$ a fixed vector and $R$ a constant, have to hold for each $\vec{\wp}\in \mathcal{B}(S)$. In this context, these are examples of linear Bell inequalities; they may be evaluated from experimentally accessible data to test whether a behaviour is in the set $\mathcal{B}(S)$ or not. Any inequality that can witness for some behaviours that $\vec{\wp}\notin \mathcal{B}(S)$ is referred to as a \emph{Bell inequality}, and the phrase \emph{optimal} or \emph{tight} Bell inequality is used for those which delimit the facets of the polytope \cite{Brunner2014}.

The set $\mathcal{NS}(S)$ is also a polytope, see for example \cite{Pironio2005} for a general proof. While the set  $\mathcal{Q}(S)$ is also convex, its geometry is significantly more complicated than a polytope, as has been explored in detail for example in \cite{Goh2017,Le2023quantumcorrelations}. 
Nonetheless, the non-strict inclusions $\mathcal{B}(S) \subset\mathcal{Q}(S) \subset \mathcal{NS}(S)$ are straightforwardly demonstrable. This combined with the facts that the quantum upper bound $R_Q$ for a given Bell inequality can be larger than the LHV bound \cite{Tsirelson1980} while the no-signalling bound $R_{NS}$ of a given inequality can be even higher than the quantum bound \cite{Popescu1994}  is sufficient to establish the well-known strict inclusions $\mathcal{B}(S) \subsetneq \mathcal{Q}(S) \subsetneq \mathcal{NS}(S)$ \cite{Brunner2014}.

Lastly, we recall that the dimension of a polytope is the dimension of its affine hull \cite{Grünbaum2003}.  Since the sets $\mathcal{B}(S), \mathcal{Q}(S)$ and $\mathcal{NS}(S)$ are all constrained by the normalization and no-signalling conditions, none of them are full-dimensional in $\mathbb{R}^D$ and in fact, it can be shown \cite{Pironio2005} that $\mathrm{dim}[\mathcal{B}(S)] = \mathrm{dim}[\mathcal{NS}(S)]$, which also means that the set $\mathcal{Q}(S)$ lives in the same affine subspace of $\mathbb{R}^D$ as the LHV and no-signalling sets.  

For our purpose, it turns out that these basic notions of polytope theory are sufficient.  We are now ready to define the main object of study of this work.

\definition{(Local Friendliness behaviours)\label{LocalFriendlinessBehaviours} \\
The set $\mathcal{LF}(S)$ of \emph{Local Friendliness behaviours} for a scenario $S= (I_A, M, O, I_F)$ consists of those behaviours that can be given a model consistent with the defining Eqs.~\eqref{JointforLF}-\eqref{LFlocalAgency2}. These equations can be used in conjunction with the definition of conditional probabilities  to decompose any behaviour $\wp(\Vec{a}|\Vec{x})$ as 
\begin{align}
    \wp(\Vec{a}|\Vec{x}) &= \sum_{\vec{c}}P^{\mathcal{LA}}(\Vec{a}|\Vec{x}, 
    \vec{c}) \cdot P(\Vec{c}) \nonumber\\
    &= \sum_{\Vec{c}}P^{\mathcal{LA}}(\Vec{a}_{F_{\Vec{x}}}| \Vec{x}, \Vec{c}, \Vec{a}_{J_{\Vec{x}}}) \cdot P^{\mathcal{LA}}(\Vec{a}_{J_{\Vec{x}}}| \Vec{x}, \Vec{c}) \cdot P(\Vec{c}) \nonumber \\
    &= \sum_{\Vec{c}}\prod_{i \in F_{\Vec{x}}} \delta_{a_i,c_i}\cdot P^{\mathcal{LA}}(\Vec{a}_{J_{\Vec{x}}}|\Vec{x}_{J_{\Vec{x}}}, \Vec{c}) \cdot P(\Vec{c}),\label{LFDefiningEquation}
\end{align}
 where $F_{\Vec{x}} = \{ i \in I_F| \vec{x}_i =1  \} $,   $J_{\vec{x}} = (I_A \setminus F_{\Vec{x}})$,
 and $P^{\mathcal{LA}}(\Vec{a}_{J_{\Vec{x}}}|\Vec{x}_{J_{\Vec{x}}}, \Vec{c})$ are distributions satisfying Local Agency (Eq.~\eqref{LFlocalAgency2}).  }

\normalfont

We will take Eq.~\eqref{LFDefiningEquation} to be the defining mathematical condition of Local Friendliness behaviours in a canonical scenario. 
Let us then proceed by presenting some basic facts that hold generally for sets $\mathcal{LF}(S)$  compatible with Def.~\ref{LocalFriendliness}.

\theorem{The set $\mathcal{LF}(S)$ is a convex polytope in any standard LF scenario. \label{LFpolytopeTheorem}}

\begin{proof}
    This is a straightforward generalization of the proof of the case in ref.~\cite{Bong2020}. We will show that any behaviour compatible with Local Friendliness may be decomposed as a finite convex sum of its extreme points.
    
By definition the behaviour decomposes as in Eq. \eqref{LFDefiningEquation}. 

Let's constrain the behaviour to the set of inputs where none of the superobservers choose a query measurement. That is, to any string $\Vec{x}'\in M' = M \setminus \{\Vec{x}|$ $ \exists k \in I_F  $ s.t. $   x_k=1 \}$. Constrained on this set of inputs the behaviour in Eq. \eqref{LFDefiningEquation} decomposes simply as 
\begin{equation}
    \wp(\Vec{a}|\vec{x}') = \sum_{\Vec{c}}
                P^{\mathcal{LA}}(\vec{a}|\vec{x}'
, \vec{c})\cdot  P(\vec{c}) \label{NosignallingDecompositionOfLF} .
\end{equation}

It is also understood that the outputs $\Vec{a}$ are therefore also constrained to $O' = \{O_{\Vec{x}'}\}_{\Vec{x}' \in M'} $. For each fixed value of $\vec{c}$, the distributions P$^{\mathcal{LA}}(\vec{a}|\vec{x}'
, \vec{c})$ can be understood as elements in the no-signalling polytope $\mathcal{NS}(S_p')$, where $S_p' = (I_A, M',O')$ and there are no further constraints.  Therefore, the distributions in \eqref{NosignallingDecompositionOfLF} can be further expanded in terms of the extreme points as

\begin{equation}\label{extremepointsDecomposition}
    P^{\mathcal{LA}}(\vec{a}|\vec{x}',\vec{c}) = \sum_{l}
                P(l|\Vec{c}) \cdot  P_{EXT}^{\mathcal{NS}(S_p')}(\vec{a}|\vec{x}', l,\vec{c}),
\end{equation}
where $l$   labels the extreme points $ P_{EXT}^{\mathcal{NS}(S_p')}(\vec{a}|\vec{x}', l,\vec{c}) \in EXT (\mathcal{NS}(S_p')) $  of said no-signalling polytope 
and $ P(l|\Vec{c})$ is the associated convex weight. 

On the other hand, since the probabilities $ P^\mathcal{LA}(\vec{a}_{J_{\Vec{x}}}|\vec{x}_{J_{\Vec{x}}}, \vec{c})$ in Eq.~\eqref{LFDefiningEquation} must arise, due to Local Agency, as the marginals of a distribution in \emph{any} context where the substring $\vec{x}_{J_{\Vec{x}}}$ is present, it means that this must be the case in particular in the maximal contexts where none of the superobservers with a friend read their outcome. Therefore, the equality 
 \begin{align}  \label{eq:PLA_Jx}
     P^{\mathcal{LA}}(\vec{a}_{J_{\vec{x}}}|\vec{x}_{J_{\vec{x}}},\vec{c})  &= \sum_{a_i, i\in F_{x}}P^{\mathcal{LA}}(\vec{a}|\vec{x}', \vec{c}) 
     \end{align}
has to hold for all $\vec{x}, F_{\vec{x}}= \{ i\in I_A | \vec{x}_i =1   \}$, $ J_{\vec{x}}= (I_A\setminus F_{\vec{x}})$ and for all $\vec{x}'\in S'$ for which $\vec{x}'_{J_{\vec{x}}} = \vec{x}_{J_{\vec{x}}} $.

From Eqs.~\eqref{extremepointsDecomposition} and \eqref{eq:PLA_Jx}, the distributions $ P^{\mathcal{LA}}(\vec{a}_{J_{\vec{x}}}|\vec{x}_{J_{\vec{x}}},\vec{c})$ can be expressed as marginals of a convex sum of extreme points so that 
     \begin{align} 
     P^{\mathcal{LA}}(\vec{a}_{J_{\vec{x}}}|\vec{x}_{J_{\vec{x}}},\vec{c})
     &= \sum_{a_i, i \in F_{\vec{x}}} \sum_{l} P(l|\Vec{c}) \cdot  P_{EXT}^{\mathcal{NS}(S_p')}(\vec{a}|\vec{x}', l,\vec{c}) \nonumber \\
     &=\sum_{l} P(l|\Vec{c}) \cdot  P_{EXT}^{\mathcal{NS}(S_p')} \label{LocalAgencyInPolytopeProofMarginalEq}(\vec{a}_{J_{\vec{x}}}|\vec{x}_{J_{\vec{x}}}, l,\vec{c}) .  
    \end{align}
    
    It is useful to define  $\mu = (l, \Vec{c}$), $P(\mu) = P(l| \Vec{c})\cdot P(\Vec{c})$, so that when Eq.~\eqref{LocalAgencyInPolytopeProofMarginalEq} is put back to Eq.~\eqref{LFDefiningEquation} any behaviour can thus be given a finite convex context-dependent decomposition in terms of 
\begin{equation}\label{Eq.LFBehaviourPolytope}
    \wp(\Vec{a}|\vec{x}) = \sum_{\mu} P(\mu)
                P_{EXT}^{\mathcal{NS}(S_p')}(\vec{a}_{J_{\Vec{x}}}|\vec{x}_{J_{\Vec{x}}}, \mu)\cdot \left( \prod_{i\in F_{\Vec{x}}} \delta_{a_i , c_i(\mu)} \right).
\end{equation}
 The proof concludes by inspection of the fact that owing to the extremality of the elements of the no-signalling polytope and the distributions determined by $\Vec{c}$, this list of points not only identifies a unique point for each $\mu$, but also it exhaustively characterizes all extreme points of $\mathcal{LF}(S)$ compatible with the decomposition of Eq.~\eqref{LFDefiningEquation}.
\end{proof}
\normalfont

This method of proving Theorem \ref{LFpolytopeTheorem} has the advantage that it immediately identifies the form of the extreme points of $\mathcal{LF}(S)$. Finding the inequalities that bound the polytope $\mathcal{LF}(S)$ may then, at least in principle, be treated as an instance of a facet enumeration problem, where the vertices of the polytope are known (or can be solved) and may then be approached computationally. 
\normalfont
As an immediate corollary of the construction of Theorem \ref{LFpolytopeTheorem}, it is evident that if $I_F = \emptyset$ then $\mathcal{LF}(S) = \mathcal{NS}(S)$. Since such constraints may never be violated by correlations arising from quantum mechanics (or any post-quantum non-signalling theories), we will henceforth refer to such scenarios as trivial LF scenarios, because the assumption of Local Friendliness does not impose any conditions that can be plausibly expected to be experimentally violated. We report this result as a theorem for ease of reference.
\theorem{If $S$ is a standard LF scenario with $I_F = \emptyset$ then $\mathcal{LF}(S) = \mathcal{NS}(S).$ \label{LFequalNScorollary}}
\begin{proof}
    This follows straightforwardly from Eq.~\eqref{Eq.LFBehaviourPolytope} combined with the facts that if $I_F =\emptyset$, then $\mathcal{NS}(S_p') = \mathcal{NS}(S_p)$, $F_{\vec{x}} = \emptyset$ and $\Vec{x}_{J_{\Vec{x}}}=\Vec{x}$,  for all $\vec{x}$.
\end{proof}
\normalfont

\theorem{ The set $\mathcal{LF}(S)$ has

\begin{align}
   & |EXT(\mathcal{LF}(S))| 
    =  \left[\prod_{i \in I_F} |O_{x_i =1}|\right]\times |EXT(\mathcal{NS}(S_p')|
\end{align} extreme points. Here $S_p' = (I_A, M', O') $ with $M' = M\setminus \{\Vec{x}|$ $ \exists k \in I_F  $ s.t. $  x_k=1 \}$ and $O' = \{ O_{\vec{x}} \}_{\vec{x} \in M'}$. \label{LFextremepointCorollary}
\begin{proof}
    Follows immediately from the form of the extreme points in Theorem \ref{LFpolytopeTheorem}. This quantity may be identified as the number of different points $\mu = \Vec{c},l$. If $I_F = \emptyset$, the extreme points of $\mathcal{LF}(S)$ match those of $\mathcal{NS}(S)$, and the empty product $\prod_{i \in I_F} |O_{x_i =1}|$ is conventionally defined to equal 1. 
\end{proof}

\normalfont

\theorem{In any standard LF  scenario $S$ the inclusions
$\mathcal{B}(S) \subset \mathcal{LF}(S) \subset \mathcal{NS}(S)$ hold.
\label{LFpolytopeisSubsetSupersetBasic}
\begin{proof}
The second inclusion holds trivially by Def.~\ref{LA} of Local Agency and the decomposition of any $\wp(\Vec{a}|\Vec{x})$ that is implied by it as in Def.~\ref{LocalFriendlinessBehaviours}. Hence we will only prove the first inclusion  $\mathcal{B}(S) \subset \mathcal{LF}(S)$.

Let $\Vec{\alpha}$ denote a string determining the outcomes of all measurements of every party as in Theorem \ref{FinesJointDistribution} and $\alpha_{x_i}(\Vec{\alpha}) \in O_{x_i}$ the outcome of the $i$th party for the input $x_i$ in the string $\vec{\alpha}$. Let $U = \{ \alpha_{x_i=1}(\Vec{\alpha}) | i\in I_F \}$ and $V= U^c =  \{ \alpha_{x_i}(\Vec{\alpha}) : i \notin I_F \textrm{ or } x_i \neq 1 \}$.
By Theorem \ref{FinesTHRM} any $\wp(\Vec{a}|\Vec{x}) \in \mathcal{B}(S)$ can then be decomposed as
\begin{align}
    \wp(\Vec{a}|\Vec{x}) &= \sum_{\Vec{\alpha}} \left( \prod_{i \in I_A} \delta_{a_i, \alpha_{x_i} (\vec{\alpha})} \right) \cdot {P}(\vec{\alpha}).\\
     &= \sum_{\alpha_{x_i}(\Vec{\alpha})\in U} \left[\sum_{\alpha_{x_i}(\Vec{\alpha}) \in V} \left( \prod_{i \in I_A}  \delta_{a_i, \alpha_{{x_i}}(\vec{\alpha})} \right) \cdot {P}(\vec{\alpha}) \label{LHVareLFsecondline}\right] \\
     \begin{split}
     &=\sum_{\alpha_{x_i}(\Vec{\alpha})\in U} \prod_{i \in F_{\Vec{x}}} \delta_{a_i, \alpha_{x_i}(\Vec{\alpha}_U)}  \\& \hspace{0.5cm} \times  P^{\mathcal{LA}}(\Vec{a}_{J_{\Vec{x}}}  | \vec{x}_{J_{\Vec{x}}}, \Vec{\alpha}_U ) {P}(\Vec{\alpha}_U), \label{LHVareLFlastline}
     \end{split}
     \end{align}
where $a_i, \Vec{a}_{J_{\Vec{x}}}, x_i$ and $ \vec{x}_{J_{\Vec{x}}}$ are elements, or substrings, of $\Vec{a} $  and $\Vec{x}$ respectively with  $F_{\Vec{x}} = \{ i \in I_F| \vec{x}_i =1   \} $ and  $J_{\vec{x}} = (I_A \setminus F_{\Vec{x}})$ as before. To get from line \eqref{LHVareLFsecondline} to \eqref{LHVareLFlastline} the summation over elements in $V$ is performed, and the product over elements in $I_A$ is split into the two context dependent subsets. 

 Equation \eqref{LHVareLFlastline} has the desired structure. For completeness,  let us identify $\alpha_{x_i=1}(\Vec{\alpha}) =c_i$ for any $i \in I_F$. Then Eq.~\eqref{LHVareLFlastline} can be rewritten in a familiar form as 
     \begin{align}
    \wp(\Vec{a}|\Vec{x}) &= \sum_{\Vec{c}} \prod_{i \in F_{\vec{x}}}\delta_{a_i , c_i}\cdot   P^{\mathcal{LA}}(\vec{a}_{J_{\Vec{x}}}|\vec{x}_{J_{\Vec{x}}}  
, \vec{c})\cdot P(\vec{c}).
\end{align}
 Here $P(\Vec{c}) = {P}(\Vec{\alpha}_U)$ and so on. This is a model that fits Def.~\ref{LocalFriendlinessBehaviours} and hence $\mathcal{B}(S) \subset \mathcal{LF}(S)$.
\end{proof}

\corollary{$\mathrm{dim} (\mathcal{B}(S)) = \mathrm{dim} (\mathcal{LF}(S)) = \mathrm{dim} (\mathcal{NS}(S)).$\label{LFdimensionCorollary}}
\begin{proof}
    Follows straightforwardly from Theorem \ref{LFpolytopeisSubsetSupersetBasic} and the fact that $\mathrm{dim}(\mathcal{B}(S)) = \mathrm{dim}(\mathcal{NS}(S))$, owing to the fact that both sets are contained in the affine subspace constrained by the normalization and no-signaling constraints \cite{Pironio2005}.
\end{proof}

\theorem{If $S = (S_p, I_F)$ and $S' = (S_p, I_F' )$ are otherwise identical LF scenarios but with $I_F' \subset I_F$, then $\mathcal{LF}(S) \subset \mathcal{LF}(S')$. \label{basicLFinclusionTHRM}}

    \begin{proof}The proof is analogous but simpler than that of Theorem \ref{LFpolytopeisSubsetSupersetBasic}. Using the properties of Local Friendliness correlations in scenario $S$, one can see that if $\wp(\Vec{a}|\Vec{x}) \in \mathcal{LF}(S)$ then,
\begin{align}
    \wp(\Vec{a}|\Vec{x}) &= \sum_{\Vec{c}} \label{FirstLineofProof}
                P^{\mathcal{LA}}(\vec{a}, \vec{c}|\vec{x})\\
                &=  \sum_{c_k, k \in I_F'} \bigg[ \sum_{c_s, s \in (I_F \setminus I_F') }  P^{\mathcal{LA}}(\vec{a}, \vec{c}|\vec{x}) \bigg]          \\
                &= \sum_{c_k, k\in I_F'}\bigg[ P^{\mathcal{LA}}(\vec{a}, \label{NotationLine}\vec{c}_{I_F'}|\vec{x}) \bigg] \\
                &= \sum_{c_k, k\in I_F'}P^{\mathcal{LA}}(\vec{a}|\vec{x}, \label{LocalAgencyLine}\vec{c}_{I_F'})\cdot P(\vec{c}_{I_F'}).
\end{align}
Here, the sum over  $\vec{c}$ in Eq.~\eqref{FirstLineofProof} has been split into two sums over a bipartition of the set of agents $I_F$, namely $I_F'$ and $(I_F \setminus I_F')$. The equality between lines \eqref{NotationLine}-\eqref{LocalAgencyLine} holds owing to Local Agency. By writing $\vec{c}_{I_F'}=\vec{c}'$, $ F_{\vec{x}}' = \{ i \in I_F', x_i =1 \}$, $J_{\vec{x}}' = I_A \setminus F_{\vec{x}}'$, and using the fact that according to the protocol if $x_i =1$ then $P(\vec{a}_i|\vec{x},\vec{c}) = \delta_{a_i,c_i}$, which implies that Eq. \eqref{LocalAgencyLine} decomposes as
    \begin{align}
        \wp(\Vec{a}|\Vec{x}) &= \sum_{\Vec{c}'} 
            P^{\mathcal{LA}}(\vec{a}_{J_{\vec{x}}'} |\vec{x}_{J_{\vec{x}}'}  
, \vec{c}')\cdot \prod_{i\in F_{\Vec{x}'} } \delta_{a_i , c_i'}\cdot P(\vec{c}'). 
\end{align}
 This is a model compatible with LF in scenario $S'$ which proves the claim.
\end{proof}
\normalfont

In the next Section, we will further strengthen Theorem $\ref{basicLFinclusionTHRM}$ by identifying all the scenarios where the inclusion is strict. 

In ref.~\cite{Bong2020} the implications of Local Friendliness in the cases of $|I_A| = |I_F| = 2 $ were investigated. There it was recognized that the constraints of Def.~\ref{LocalFriendlinessBehaviours} define a convex polytope, for which in general it holds that $\mathcal{B}(S) \subset \mathcal{LF}(S) \subset \mathcal{NS}(S)$. The cases where the superobservers have  two settings with arbitrary number of outputs each and three settings with two outputs each were solved completely. In the case of two inputs the LF polytope equals the Bell polytope. In the case of three inputs the LF polytope does not equal the Bell polytope, leading to the notion of ``genuine Local Friendliness inequalities" \cite{Bong2020}, that is, facet-defining LF inequalities that are not facet-defining Bell inequalities for the same scenario. 

It also turns out (see acknowledgements of ref.~\cite{Bong2020}) that similar geometric objects had been previously independently investigated in the appendix of the PhD thesis of Erik Woodhead \cite{Woodhead2014} under the name of ``partially deterministic polytopes'', where they emerged as conditions for certifying randomness in the presence of an adversary constrained by no-signalling. 
As a corollary\footnote{The bipartite partial deterministic polytopes of \cite{Woodhead2014} are slightly more general than the LF polytopes in canonical scenarios. In \cite{utrerasalarcón2023} it was shown that the more general partially deterministic polytopes can be the objects of interest for Local Friendliness in the so-called sequential scenarios. In forthcoming work \cite{HaddaraWisemanCavalcanti} we will explore partially deterministic polytopes and some of their physical significance in full generality. For now, one may think of the polytopes $\mathcal{LF}(S)$ of this work as examples of specific forms of partial determinism, and among our main novel results are the mathematical properties of multipartite extensions of these objects.} of Woodhead's results \cite{Woodhead2014}, $\mathcal{B}(S)=\mathcal{LF}(S)$ in any bipartite canonical scenario $S$ where one of the parties has a friend and only two settings, that is, if $\exists i \in I_F$ s.t.~$|M_i|=2$ (regardless of the number of settings of the other party). 

Beyond the investigations from refs.~\cite{Bong2020,Woodhead2014}, we are not aware of results on general behaviours compatible with Local Friendliness. In particular, it seems that, aside from the tripartite scenario with $I_F=I_A=3$ and two dichotomic measurements per party, discussed in ref.~\cite{Ding2023}\footnote{In ref.~\cite{Ding2023}, a Bell inequality for that scenario is claimed, without rigorous proof, to follow from Local Friendliness. We return to this matter  in Sec. \ref{DiscussionSection}.}, 
the constraints implied by Local Friendliness in multipartite scenarios have so far been left unexplored. In Theorems \ref{LFpolytopeTheorem}-\ref{LFpolytopeisSubsetSupersetBasic}, \ref{basicLFinclusionTHRM} and Corollary \ref{LFdimensionCorollary}, we have established basic facts that hold for the sets $\mathcal{LF}(S)$ in arbitrary canonical LF scenarios, including the fact that these sets may always be viewed as convex polytopes.

 As is known from cases of solving the inequalities bounding the Bell polytopes, e.g.~\cite{Pitowsky1989, Pitowsky2001} and the extreme points of the no-signalling polytopes \cite{Pironio2011}, finding complete solutions to these problems quickly becomes computationally hard as the number of parties, inputs and outputs are increased. Combined with the fact that there are cases where $\mathcal{LF}(S)$ equals the LHV polytope $\mathcal{B}(S)$, this motivates us to analytically explore the polytopes $\mathcal{LF}(S)$, identifying in which cases they equal the corresponding Bell polytope and cases where new inequalities might be found, thereby pruning the computational search problem. In the next Section we present the main results of this work, which resolve this question completely.

\section{Results: Characterization of LF polytopes \label{ResultsSection}}

\begin{theorem}\label{N-input-forAlice-two-inputOthersThreom1Charliemissing} If in a canonical Local Friendliness scenario $|I_F|=|I_A|-1$ and  $\forall j \in I_F$,  $|M_j| = 2$, then $ \mathcal{LF}(S) = \mathcal{B}(S)$.
\end{theorem}

\begin{proof} Suppose, without loss of generality, that the first superobserver does not have a friend and has $n$ measurement settings $x_1\in\{1 \ldots n\}$, while each of the N-1 remaining superobservers has a friend and 2 measurement settings with labels $x_{i\neq 1}\in\{1,2\}$. In any given run of the experiment, given the string of inputs $\vec{x}$, let $F_{\vec{x}}  = \{ i \in I_F| \vec{x}_i =1  \} $  represent the context-dependent subset $F_{\vec{x}} \subset I_F =$ ($I_A \setminus \{1 \})$ of superobservers with a friend who choose $x_i = 1$. Then any behaviour $\wp(\Vec{a}|\Vec{x})$ with $\Vec{x} \in \{1 \ldots n \} \times \{1,2\}^{N-1}$ consistent with Local Friendliness may be decomposed as
\begin{align}
    \wp(\Vec{a}| {\vec{x}}) =\sum_{\Vec{c}} P^{\mathcal{LA}} (\vec{a}_{J_{\vec{x}}}|\vec{x}_{J_{\vec{x}}}, \vec{c})\cdot \prod_{i\in F_{\vec{x}}}\delta_{a_i , c_i} \cdot P(\vec{c}), \label{LHVproofdecomposition}
\end{align}
 with  $J_{\vec{x}} = (I_A \setminus F_{\Vec{x}})$.
 
 Note that the deterministic distributions $\prod_{i\in F_{\vec{x}}}\delta_{a_i , c_i}$ are already consistent with a LHV model, as phrased in Def.~\ref{LHVbehaviours}, for any $\vec{x}$. On the other hand, owing to Local Agency, the distributions  $P^{\mathcal{LA}}(\vec{a}_{J_{\vec{x}}}|\vec{x}_{J_{\vec{x}}}, \vec{c})$ in Eq.~\eqref{LHVproofdecomposition} are to be recovered as marginals from distributions $P^{\mathcal{LA}}(\Vec{a}|\Vec{x}, \Vec{c})$ in the contexts where $J_{\vec{x}} = I_A$, that is where  $\Vec{x} \in \{1, \ldots , n\} \times \{2\}^{N-1}$. Therefore, as shown in detail below,  if this set of $n$ distributions for each $\vec{c}$ can be given an LHV model then the whole behaviour $ \wp(\Vec{a}| {\vec{x}})$ has an LHV model. 
 
 By Theorem \ref{FinesTHRM}, the existence of an LHV model for the set of n distributions specified above is equivalent to the existence of a joint probability distribution $P(\vec{\alpha}|\vec{c})$ which recovers the distributions $P^{\mathcal{LA}}(\Vec{a}|\Vec{x}, \Vec{c})$ with  $\Vec{x} \in \{1, \ldots , n\} \times \{2\}^{N-1}$ as marginals:
 
 \begin{equation}
\begin{aligned}
   P(\vec{\alpha}_{\vec{x}} = \vec{a}_{\vec{x}}|\vec{c}) = P^{\mathcal{LA}}(\Vec{a}|\Vec{x}, \Vec{c}),
    \end{aligned}
\end{equation} with with $P(\vec{\alpha}_{\vec{x}}|\vec{c}) = \sum_{\alpha_{x_i'}: x_i' \neq x_i} {P}( \Vec{\alpha}|\vec{c})$, in analogy with the convention of Theorem \ref{FinesJointDistribution}.

We will show that such a joint distribution exists.

It will be useful to employ the notation $P^{\mathcal{LA}}(\Vec{a}|\Vec{x}, \Vec{c}) = P^{\mathcal{LA}}(a_{x_1=s}, a_{x_2=2} \ldots, a_{x_N=2)} |\Vec{c})$, with $s \in \{1, \ldots n \}$ corresponding to $\Vec{x} \in \{1, \ldots , n\} \times \{2\}^{N-1}$.  Now due to Local Agency, marginalizing over superobserver 1 in $P^{\mathcal{LA}}(\Vec{a}|\Vec{x}, \Vec{c})$ produces the same distribution  $P^{\mathcal{LA}}( a_{x_2=2}, \ldots , a_{x_N=2}| \Vec{c})$ for all $x_1$. Inspired by the construction presented in the appendix of ref.~\cite{Wolf2009} (see also \cite{Fine1982b}), consider then the assembly

\begin{equation}
\begin{aligned} \label{massiveJoint}
     P(\alpha_{x_1=1},&\alpha_{x_1=2} , \ldots , \alpha_{x_1=n} , \alpha_{x_2= 2}, a_{x_3=2}, \ldots ,\alpha_{x_N=2}| \vec{c}) \\ :=& {P}(\Vec{\alpha}|\vec{c}) \\  =& \dfrac{\prod_{s=1}^{n} P^{\mathcal{LA}}(a_{x_1=s}, a_{x_2=2}, a_{x_3=2}, \ldots, a_{x_N=2 } |\Vec{c})}{[P^{\mathcal{LA}}( a_{x_2=2}, \ldots , a_{x_N=2} | \Vec{c})]^{n-1}} \textrm.
    \end{aligned}
\end{equation} 

Here ${P}(\Vec{\alpha}|\vec{c}) $ is set to zero whenever the divisor on the right hand side equals zero.

Equation \eqref{massiveJoint} defines a valid probability distribution because 

1) ${P}(\Vec{\alpha}|\vec{c}) \geq 0 $ $\forall \vec{\alpha}$, since it is defined by multiplying and dividing probabilities. \\

2) $\sum_{\vec{\alpha}} {P}(\Vec{\alpha}|\vec{c}) =1.$ This may perhaps be easiest seen by considering summation over the outcomes of all settings for the first party: 
\begin{equation}
    \begin{aligned}
        \sum_{a_{x_1=1},\ldots , a_{x_1=n}} {P}(\Vec{\alpha}|\Vec{c}) &  \equiv \dfrac{ [P^{\mathcal{LA}}( a_{x_2=2}, \ldots , a_{x_N=2} | \Vec{c})]^{n}}{[P^{\mathcal{LA}}( a_{x_2=2}, \ldots , a_{x_N=2} | \Vec{c})]^{n-1}} \\ &= P^{\mathcal{LA}}( a_{x_2=2}, \ldots , a_{x_N=2)} | \Vec{c})
    \end{aligned}
\end{equation}
where the equivalence following the summation holds true by appeal  to Local Agency, and hence summing over remainining agents amounts to 1. 

In an analogous manner one may confirm that the condition leading to Fine's theorem of being able to recover the appropriate $P^{\mathcal{LA}}(\Vec{a} |\Vec{x},  \Vec{c}) $ for any $\vec{x} \in \{ 1 \ldots n \} \times \{ 2\}^{N-1}$ from $ {P}(\Vec{\alpha}|\vec{c})$ as marginals is satisfied. Namely, by summing over the inputs of superobserver 1 not corresponding to the choice $s$ in a particular $\Vec{x}$.

The proof concludes by noting that the behaviour $\wp(\vec{a}|\vec{x})$ can now be admitted a deterministic LHV expansion of the kind as in Theorem \ref{FinesDeterministicDistribution} by setting  $\lambda =(\vec{\alpha}, \vec{c}),  P(\lambda) := {P}(\Vec{\alpha}|\vec{c}) P(\Vec{c})$. This can  now be understood as a distribution over LHV's which determine the outcomes of the entire experiment, that is one may express
\begin{equation}
    \begin{aligned}
        \wp(\Vec{a}| {\Vec{x}}) = \sum_{\lambda} \prod_{i \in I_A }\delta_{a_i, \alpha_{x_i}(\lambda) } P(\lambda) \label{LHVforLFequation}
    \end{aligned}
\end{equation}
with $\alpha_{x_i}: \Lambda \rightarrow O_{x_i}$  so that if $i\in I_F$ and $x_i =1$  the outcome $a_i(\vec{\alpha}, \vec{c})$ is determined  by $\Vec{c}$, and in every other case by  $\Vec{\alpha}$.
\end{proof}

\normalfont
As it turns out, Theorem \ref{N-input-forAlice-two-inputOthersThreom1Charliemissing} is the most general positive result that may be established with regards to the existence of LHV models for behaviours compatible with Local Friendliness. In Section \ref{DiscussionSection} we provide simple examples of this construction in action. For now, we proceed with exploring the LF polytopes. Before completing the picture with the converse (or negative) results, we establish an immediate corollary of this theorem.

\corollary{\label{N-input-forAlice-two-inputOthersThreom} In any canonical LF scenario, if $I_A = I_F$ and  $|M_i| = 2$ for $i$ ranging in at least $|I_A|-1 $ parties, then $ \mathcal{LF}(S) = \mathcal{B}(S)$.}
\begin{proof} Follows immediately from Theorems \ref{N-input-forAlice-two-inputOthersThreom1Charliemissing}, \ref{basicLFinclusionTHRM} and \ref{LFpolytopeisSubsetSupersetBasic}. For completeness, we demonstrate how an LHV model may be built analogously to the case of Theorem \ref{N-input-forAlice-two-inputOthersThreom1Charliemissing}.
    The construction follows the recipe of Theorem \ref{N-input-forAlice-two-inputOthersThreom1Charliemissing} so we may be substantially briefer here. Suppose again that agent 1 has $n$ inputs. The subset $F_{\vec{x}}$ in the proof of Theorem \ref{N-input-forAlice-two-inputOthersThreom1Charliemissing} may now be any subset of $I_A$ and hence it is sufficient to provide an LHV model for the $n-1$ distributions of the form $P^{\mathcal{LA}}(\Vec{a}|{\Vec{x}}, \Vec{c})$, where now $\Vec{x} \in \{2, \ldots , n\} \times \{2\}^{N-1}$. This is a smaller set than in the case of Theorem \ref{N-input-forAlice-two-inputOthersThreom1Charliemissing} though the scenario is otherwise identical. Consider then the joint distribution 
        \begin{equation}
\begin{aligned} \label{massiveJoint2}
    P(\vec{\alpha}|\vec{c})   := \dfrac{\prod_{s=2}^{n} P^{\mathcal{LA}}(a_{x_1=s}, a_{x_2=2}, a_{x_3=2 } \ldots a_{x_N=2 } |\Vec{c})}{[P^{\mathcal{LA}}( a_{x_2,2}, \ldots , a_{x_N=2} | \Vec{c})]^{n-2}},
    \end{aligned}
\end{equation} 
   which can be verified to recover the correct distributions $P^{\mathcal{LA}}(\Vec{a}|{\Vec{x}}, \Vec{c})$ with $\Vec{x} \in \{2, \ldots , n\} \times \{2\}^{N-1}$ as marginals, analogously to the proof of Theorem \ref{N-input-forAlice-two-inputOthersThreom1Charliemissing}. Then, analogously to Eq.~\eqref{LHVforLFequation} one can identify $\lambda = (\vec{\alpha}, \vec{c})$, $P(\lambda) = P(\vec{\alpha}|\vec{c})P(\vec{c})$ and write 

   \begin{align}
       \wp(\Vec{a}| {\Vec{x}}) = \sum_{\lambda} \prod_{i \in I_A }\delta_{a_i, \alpha_{x_i}(\lambda) } {P} (\lambda),
   \end{align}
   with $\alpha_{x_i}: \Lambda \rightarrow O_{x_i}$. 
\end{proof}

\normalfont

Intuitively Corollary \ref{N-input-forAlice-two-inputOthersThreom}  is a special case of Theorem \ref{N-input-forAlice-two-inputOthersThreom1Charliemissing} because the modification to the statement brings more constraints that bring the LF behaviour already closer to an LHV model.

\lemma{\label{ReducedPartiesLHV}Let $S'$ be a scenario that can be obtained from scenario $S$ by excluding some subset $E\subset I_A$ of parties.  If $\wp^{S}(\vec{a}|\vec{x}) \in \mathcal{NS}(S)$, then the behaviour $\wp^{S'}(\vec{a}'|\vec{x}')$ obtained from $\wp^{S}(\vec{a}|\vec{x})$ by the mapping
\begin{align}
\wp^{S}(\vec{a}|\vec{x}) \mapsto \wp^{S'}(\vec{a}'|\vec{x}')=\sum\limits_{a_i, i\in E} \wp^{S}(\Vec{a}|\Vec{x}) 
\end{align}
is a valid No-signaling behaviour in scenario $S'$. Furthermore if $\wp^{S}(\Vec{a}|\Vec{x})\in \mathcal{B}(S)$ then $\wp^{S'}(\vec{a}'|\vec{x}') \in \mathcal{B}(S').$
}
\begin{proof}
    Clearly $\wp^{S'}(\vec{a}'|\vec{x}') \in \mathcal{NS}(S')$, since the marginals are well defined owing to no-signaling. Therefore it is sufficient to show that $\wp^{S}(\vec{a}|\vec{x}) \in \mathcal{B}(S) \Rightarrow 
 \wp^{S'}(\vec{a}'|\vec{x}') \in \mathcal{B}(S')$. This is seen simply by noting that by hypothesis  $\wp^{S}(\vec{a}|\vec{x})$ admits an LHV expansion as in Eq.~ \eqref{Bell-localmaths} 
 \begin{align}
     \wp^{S}(\vec{a}|\vec{x}) = \int_{\Lambda} d\lambda \, p(\lambda) \left[\prod_{i=1}^N P(a_i|x_i, \lambda)\right].
 \end{align}
 Therefore, owing to normalization of the $P(a_i|x_i,\lambda)$ above, one gets an LHV model for $\wp^{S'}(\vec{a}'|\vec{x}')$ induced by the original one for the scenario $S'$ as well.
\end{proof}

\lemma{\label{ReducedMeasurementsLHV}Let $S'$ be a scenario that can be obtained from scenario $S$ by restricting the measurements $x_i$ of the parties $i\in I_A$ to some non-empty subsets $M_i' \subset M_i$. Then the behaviour $\wp^{S'}(\vec{a}'| \vec{x}')$ obtained from $\wp^{S}(\vec{a}|\vec{x})$ by the  `identity mapping'
\begin{align}
    \wp^{S'}(\vec{a}'| \vec{x}') = \wp^{S}(\vec{a}|\vec{x}) \hspace{0.5cm} \mathrm{iff} \hspace{0.5cm} \delta_{\vec{a},\vec{a}'}\cdot \delta_{\vec{x}, \vec{x}'}=1,
\end{align}
is a valid behaviour in scenario $S'$. Here $\vec{x}' \in M'$  is understood to represent the strings of inputs the elements of which are restricted to taking values only in the $M_i'$'s. If also $\wp^{S}(\vec{a}|\vec{x}) \in \mathcal{B}(S)$, then $\wp^{S'}(\vec{a}'|\vec{x}') \in \mathcal{B}(S')$.
}
\begin{proof}
    Clearly $\wp^{S'}(\vec{a}'|\vec{x}')$ obtained this way is a valid behaviour in $S'$. To see  that $\wp^{S}(\vec{a}|\vec{x}) \in \mathcal{B}(S) \Rightarrow \wp^{S'}(\vec{a}|\vec{x}') \in \mathcal{B}(S')$, it is sufficient to note that by  restricting the behaviour $\wp^{S}(\vec{a}|\vec{x})$ to only include the distributions conditioned on subsets of inputs $\vec{x}' \in M'$, one is essentially looking for an LHV model to a strictly smaller set of distributions than in the original case. Hence, the existence of an LHV model for the behaviour in scenario $S$, guarantees one for $S'$. 
\end{proof}
\normalfont

Lemmas \ref{ReducedPartiesLHV} and \ref{ReducedMeasurementsLHV} are admittedly trivial. However, it turns out that they can be used to establish some facts about the sets $\mathcal{LF}(S)$ in general. Let us first use them to point out yet another way to prove that Corollary \ref{N-input-forAlice-two-inputOthersThreom}  indeed follows from Theorem \ref{N-input-forAlice-two-inputOthersThreom1Charliemissing}. Namely, if $\wp(\Vec{a}|\Vec{x})$ is any behaviour in a scenario with $I_A = I_F$ where $M_i = \{1,2 \}$ for all but one $i$, say $i=1$ for which $M_1 = \{1 \ldots n \}$, then $\wp(\Vec{a}|\Vec{x})\cdot \wp (a_{N+1}|x_{N+1}) $ with $x_{N+1} \in M_{N+1} = \{1 , 2\}$ may be thought of as a valid behaviour in a scenario compatible with the premise of Theorem \ref{N-input-forAlice-two-inputOthersThreom1Charliemissing}, and hence admits an LHV model. By lemma \ref{ReducedPartiesLHV}, by marginalizing over superobserver $N+1$, so does $\wp(\Vec{a}|\Vec{x})$ as well. 

\normalfont
More than that, we can use lemmas \ref{ReducedPartiesLHV}-\ref{ReducedMeasurementsLHV} aided with simple observations to establish quite effortlessly our previous claim that the conditions of Theorem \ref{N-input-forAlice-two-inputOthersThreom1Charliemissing} are in fact necessary for the existence of an LHV model. 

\theorem{ \label{GenuineLFgenerally} The scenarios covered by Theorem \ref{N-input-forAlice-two-inputOthersThreom1Charliemissing} and Corollary \ref{N-input-forAlice-two-inputOthersThreom} are the only ones where  $\mathcal{LF}(S) = \mathcal{B}(S)$. }
\begin{proof}
   Notably, the situations in which the assumptions leading to Theorem \ref{N-input-forAlice-two-inputOthersThreom1Charliemissing} or Corollary \ref{N-input-forAlice-two-inputOthersThreom}  can fail may be grouped into three distinct classes;$|I_A| - |I_F| \geq 2$, two or more agents have $|M_i| \geq 3$ and  $|I_F|=|I_A| - 1$ and all but one superobserver have two measurements  but the superobserver with $n$ inputs has a friend.   We approach the proof by case by case analysis, though the argument is similar in each instance.    \label{NoLHVtheorem}
    \begin{enumerate}[label=\roman*), ref=\ref{NoLHVtheorem}.\roman* ]
        \item $|I_A| - |I_F| \geq 2$ \label{NoLHVtheoremCase1}

        Suppose $\wp(\Vec{a}|\vec{x}) \in \mathcal{LF}(S)$ and agents $s,t \in I_A$ are without a friend. Then, by hypothesis the behaviour admits a decomposition of the form 
        \begin{equation}
            \wp(\Vec{a}|\vec{x}) = \sum_{\vec{c}}
                P^{\mathcal{LA}}(\vec{a}_{J_{\vec{x}}}|\vec{x}_{J_{\vec{x}}}, \vec{c})\cdot \prod_{i \in F_{\vec{x}}} \delta_{a_i , c_i}\cdot P(\vec{c}) ,\label{standardLFdecompositioninTheorem}
        \end{equation}
       with  $F_{\vec{x}}  = \{ i \in I_F| \vec{x}_i =1  \} $ and $J_{\vec{x}} = I_A \setminus F_{\vec{x}}$. In particular  $s,t \notin I_F$ so that $x_s, x_t$ are always in the string $\vec{x}_{J_{\vec{x}}}$ in Equation \eqref{standardLFdecompositioninTheorem}. This means that the family of behaviours which expand, for $K_{\vec{x}} = J_{\vec{x}} \setminus \{s,t \}$, as
        \begin{equation}\label{ExampleofGenuineLFbehaviour}
        \begin{split}
            \wp(\Vec{a}|{\Vec{x}}) &= \sum_{\Vec{c}} P^{\mathcal{LA}}(\Vec{a}_{K_{\vec{x}}}|\Vec{x}_{K_{\vec{x}}}, \Vec{c})\cdot P^{\mathcal{LA}}(a_s, a_t |x_s,x_t) \\
            & \hspace{0.5cm} \times \prod_{i \in F_{\vec{x}}} \delta_{a_i , c_i}\cdot P(\vec{c}),
            \end{split}
        \end{equation}
         are perfectly compatible with the model of Eq.~\eqref{standardLFdecompositioninTheorem}. Now let us suppose, for the sake of demonstrating a contradiction, that the behaviour in Eq.~\eqref{ExampleofGenuineLFbehaviour} were LHV modelable. Then, by appeal to Lemma \ref{ReducedPartiesLHV}, the marginalized set of distributions $P^{\mathcal{LA}}(a_s, a_t |x_s,x_t)$ would also be LHV modelable. We could stop here, since it is well known that $|M_s|, |M_t| \geq 2$ combined with Local Agency is not sufficient to guarantee LHV modelability. For completeness, we demonstrate this fact with an example. Suppose that for all inputs $x_s,x_t$ there is only a pair of outcomes, $a_t,a_s \in \{1,2\}$ say,  with nonzero probability. Then 
         \begin{align}\label{LAexpansiontoPR}
         \begin{split}
             P^{\mathcal{LA}}(&a_s, a_t|x_s,x_t)  \\ 
             & = \begin{cases}
                 P_{PR}^{\mathcal{LA}}(a_s, a_t |x_s,x_t) & a_s,a_t,x_s, x_t \in \{1,2\}\\
                 P_{PR}^{\mathcal{LA}}(a_s|x_s)\cdot P_{PR}^{\mathcal{LA}}(a_t|x_t) & \mathrm{otherwise},
             \end{cases}
             \end{split}
         \end{align}
         is an example of a valid expansion of $P^{\mathcal{LA}}(a_s, a_t |x_s,x_t).$
         Here $P_{PR}^{\mathcal{LA}}(a_s, a_t |x_s,x_t)$ refers to the Popescu-Rohrlic- type correlations \cite{Popescu1994} which are 
        non-signalling but are neither LHV- or quantum modellable. By Lemma \ref{ReducedMeasurementsLHV}, Eq.~\eqref{LAexpansiontoPR} can be further mapped to contain just the inputs $x_s,x_t \in \{1,2\}$, preserving LHV modelability.  The proof concludes by noting that if an LHV model did exist for all behaviour satisfying Eq.~\eqref{standardLFdecompositioninTheorem}, then by Lemmas \ref{ReducedPartiesLHV} and \ref{ReducedMeasurementsLHV} so, in particular,  would an LHV model for $P_{PR}^{\mathcal{LA}}(a_s, a_t |x_s,x_t)$, which is a contradiction.   
        
        \item Two or more agents have $|M_i| \geq 3$ \label{NoLHVtheoremCase2}

        This may be proved analogously to the previous case of Theorem \ref{NoLHVtheoremCase1}. Now nothing prevents some pair of agents $(s,t)$ of having PR-correlations for some pairs of their measurements $x_s, x_t \neq 1$, leading to contradiction.
        
        \item \label{NoLHVtheoremCase3} $|I_F|=|I_A| - 1$ and all but one superobserver have two measurements  but the superobserver with $n$ inputs has a friend. 

        Again, now the agent $s$ with $n$ measurements and a friend may share PR correlations for some pair of her measurements $x_s \neq 1$ with the agent $t$ without a friend. 
    \end{enumerate}
\end{proof}

\normalfont
The proof of Theorem \ref{NoLHVtheorem} although simple, can be used to extract information about the relation of the sets $\mathcal{Q}(S)$ and $\mathcal{LF}(S)$ as well. Namely, because similar statements as in lemmas \ref{ReducedPartiesLHV} and \ref{ReducedMeasurementsLHV}   hold trivially on the quantum set $\mathcal{Q}(S)$ as well, it means that by an adaptation of the proof of Theorem \ref{NoLHVtheorem} one may establish that in general $\wp(\Vec{a}|\Vec{x}) \in \mathcal{LF}(S) \nRightarrow  \wp(\Vec{a}|\Vec{x})  \in \mathcal{Q}(S)$,  except precisely in the class of scenarios covered by Theorem \ref{N-input-forAlice-two-inputOthersThreom1Charliemissing} and corollary \ref{N-input-forAlice-two-inputOthersThreom} where $\mathcal{LF}(S) = \mathcal{B}(S)$. In Section \ref{DiscussionSection} we show a simple argument which demonstrates that whenever $I_F \neq \emptyset$, then also $\wp(\Vec{a}|\Vec{x}) \in \mathcal{Q}(S) \nRightarrow \wp(\Vec{a}|\Vec{x}) \in \mathcal{LF}(S)$, which means that the sets $\mathcal{Q}(S)$ and $\mathcal{LF}(S)$ have in general disjoint regions, so that neither set is a subset of the other. The two extreme cases which make an exception are the scenarios covered by Theorem \ref{N-input-forAlice-two-inputOthersThreom1Charliemissing} where $\mathcal{B}(S) = \mathcal{LF}(S) \subsetneq \mathcal{Q}(S)$ and the scenarios where $I_F = \emptyset$ where $\mathcal{Q}(S) \subsetneq \mathcal{NS}(S) = \mathcal{LF}(S)$.

The construction used in the proof of Theorem \ref{NoLHVtheoremCase1} will be useful for demonstrating our following conclusions as well, and we will henceforth be appealing to such examples with less rigour. Armed with the results presented so far, we are now also able to say something about the relations of the sets of behaviours $\mathcal{LF}(S)$ and $\mathcal{LF}(S')$ in cases where $I_A = I_{A'}, M = M'$ and $\{O_{\Vec{x}} \} = \{O_{\Vec{x}'}\}$.

\theorem{\label{LFinclusionTheorem}If $S=(S_p, I_F )$ and $S'=(S_p,I_F')$ are otherwise identical scenarios but where $\emptyset \neq I_F' \subsetneq I_F$ then either
$\mathcal{LF}(S) = \mathcal{LF}(S') = \mathcal{B}(S_p)$ or $\mathcal{LF}(S) \subsetneq \mathcal{LF}(S')$ }
\begin{proof}

    If the conditions of Theorem \ref{N-input-forAlice-two-inputOthersThreom1Charliemissing} or Corollary \ref{N-input-forAlice-two-inputOthersThreom} hold for $S$ and $S'$, then clearly the first claim holds true and $\mathcal{B}(S_p)= \mathcal{LF}(S) = \mathcal{LF}(S')$.  On the other hand  by Theorem \ref{NoLHVtheorem} the second claim holds if the conditions of Theorem \ref{N-input-forAlice-two-inputOthersThreom1Charliemissing} or Corollary \ref{N-input-forAlice-two-inputOthersThreom} are satisfied in scenario $S$ but not in scenario $S'$. We will show that also in all other cases not covered by these two  $\mathcal{LF}(S) \subsetneq \mathcal{LF}(S')$.

    The inclusion $\mathcal{LF}(S) \subset \mathcal{LF}(S')$ holds by Theorem \ref{basicLFinclusionTHRM}. It is therefore sufficient to show that the set  $\mathcal{LF}(S')$ always contains behaviour that are not in $\mathcal{LF}(S)$.
      Let us first consider the case where $|I_F|< N$. In these cases, there exists at least one pair $s,t \in I_A$ such that $s,t \notin I_F'$ while simultaneously $s \in I_F, t \in t \notin I_F$ or $t \in I_F, s \notin I_F$ holds true. Suppose without loss of generality that the first condition holds. Consider the restriction of measurements in the sense of lemma \ref{ReducedMeasurementsLHV} of these two parties to include only two inputs, say $x_s , x_t \in \{1,2 \}$ per site. Now the conditions of Theorem \ref{N-input-forAlice-two-inputOthersThreom1Charliemissing} hold for this two-party sub-behaviour in the scenario with $I_F$, and so admits an LHV model. Meanwhile, analogously to the construction of Theorem \ref{NoLHVtheorem}, nothing prevents the parties $s,t$ of sharing PR-correlations in the scenario $I_F'$. This is a contradiction since it means that the set $\mathcal{LF}(S')$ contains behaviours for which the sub-behaviours are not consistent with demands on sub-behaviours of $\mathcal{LF}(S)$. It is easy to see that a similar argument can be constructed whenever $|I_F| - |I_F'| \geq 2$. This means that the only remaining cases consists of those where $|I_F| = N$ and $|I_F'| = N-1$ and the other conditions of Theorem \ref{N-input-forAlice-two-inputOthersThreom1Charliemissing} and Corollary \ref{N-input-forAlice-two-inputOthersThreom} fail. In particular, because $|I_F| = N$, it means that there must be at least one pair $s,t$ such that $|M_s| , |M_t| \geq 3$. We may narrow our attention to two type of cases; either $s,t \in I_F'$ or $s\in I_F', t \notin I_F'$. In the first case, there exists an $r$ such that $r\notin I_F'$ but by hypothesis $r \in I_F$. Hence  the scenario $I_F'$ allows the parties to share PR - correlations  for settings $x_{s/t} \in \{2,3 \}$ and $x_r 
      \in \{1,2 \}$ while in the scenario $I_F$, restriction to this pair of settings would still be compatible with the premiss of Theorem \ref{N-input-forAlice-two-inputOthersThreom1Charliemissing} and hence admit an LHV model. The second case is analogous; but with $x_s \in \{2,3 \}$ and $ x_t \in \{1,2 \},$ say. 
\end{proof}

\normalfont
Theorem \ref{LFinclusionTheorem}  shows that the LF polytopes form a (usually strict) hierarchy $\mathcal{LF}(S_p,I_F) \subset \mathcal{LF}(S_p,I_F') \subset \ldots \subset \mathcal{LF}(S_p,I_F'') $, whenever $ I_F'' \subsetneq \ldots \subsetneq I_F' \subsetneq I_F$, and shows that in general one needs to derive new sets of inequalities for scenarios with different numbers or friends. The notable exceptions to this are precisely the cases with $|I_F| =N, |I_F'| = N-1$ and the restrictions to the measurements as in Theorem \ref{N-input-forAlice-two-inputOthersThreom1Charliemissing} and Corollary \ref{N-input-forAlice-two-inputOthersThreom} where the polytopes admit an LHV model are satisfied. On the other extreme, the result strengthens the statement of Theorem \ref{LFequalNScorollary} to include the 'only if', that is $\mathcal{LF}(S_p,I_F) = \mathcal{NS}(S_p)$ iff $I_F = \emptyset.$ The proof also points out to the intuitive fact that the 'partial determinism' imposed by the assumptions regarding the friends measurements bring the polytope in a sense closer to the Bell-polytope by demanding more and more sub-behaviour to be compatible with an LHV model.

\begin{figure}
    \centering
    \includegraphics[width=\columnwidth]{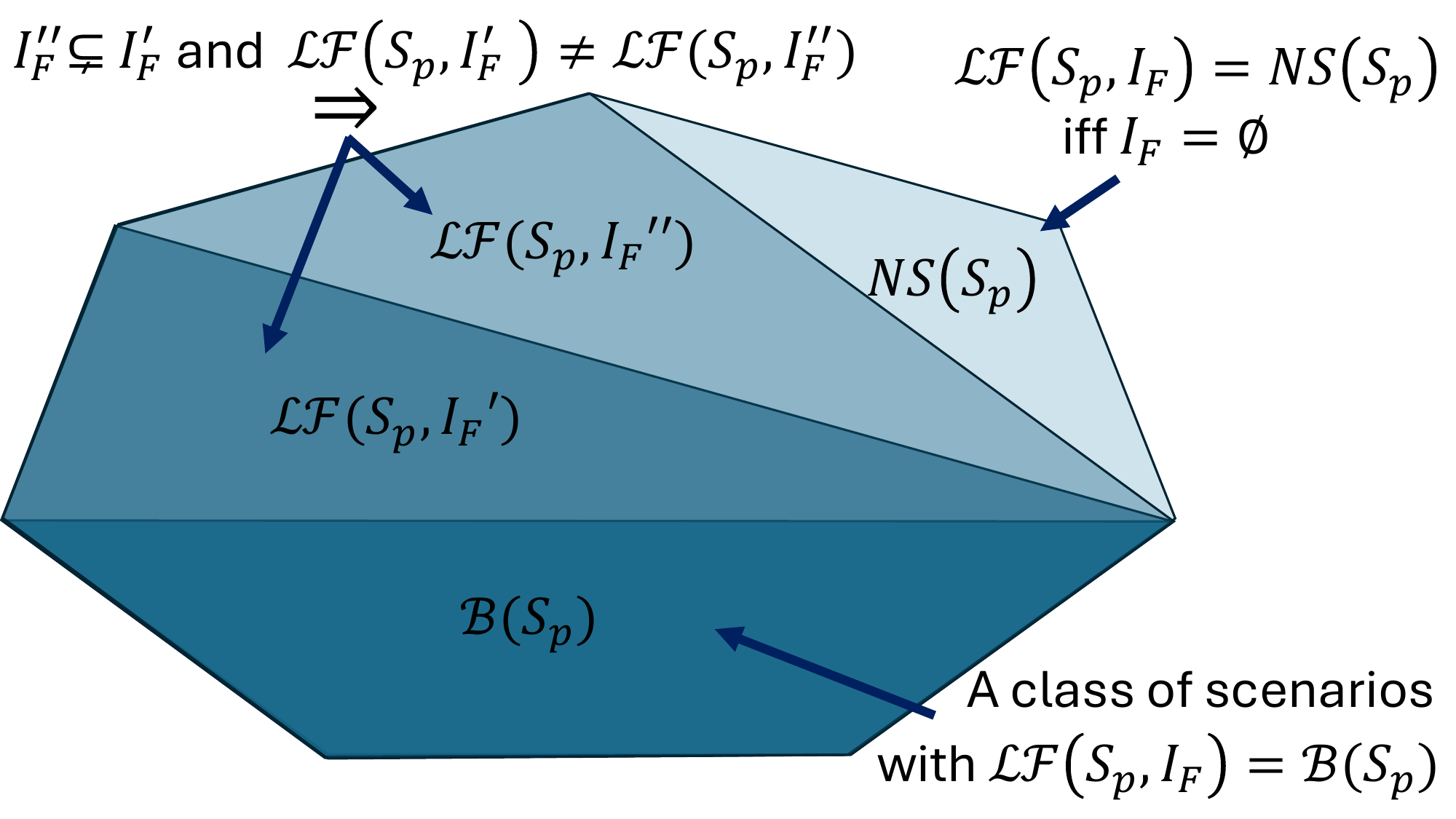}
    \caption{\label{fig:Firstresultsfig}
    An illustration of the geometry relating the basic sets $\mathcal{B}(S_p)$ and $\mathcal{NS}(S_p)$ to $\mathcal{LF}(S_p,I_F)$ in arbitrary canonical Local Friendliness scenarios. By Theorems \ref{LFpolytopeTheorem}-\ref{basicLFinclusionTHRM} the set $\mathcal{LF}(S_p,I_F)$ is always a convex polytope, which contains the Bell polytope $\mathcal{B}(S_p)$ and is contained in the no-signaling polytope $\mathcal{NS}(S_p).$ Theorem \ref{N-input-forAlice-two-inputOthersThreom1Charliemissing}, Corollary \ref{N-input-forAlice-two-inputOthersThreom} and Theorem \ref{GenuineLFgenerally} completely classify the equivalence class of scenarios where $\mathcal{LF}(S_p,I_F) = \mathcal{B}(S_p)$ hold, with constraints both on the set of friends $|I_F|\geq |I_A|-1$ and the set of inputs $M$. In particular, there are scenarios where $I_F' \subsetneq I_F''$ but $\mathcal{LF}(S_p,I_F')=\mathcal{LF}(S_p,I_F'')$, namely those where both equal the Bell polytope. Theorem \ref{LFinclusionTheorem} shows that in all other cases, where neither $\mathcal{LF}(S_p,I_F')$ or $\mathcal{LF}(S_p,I_F'')$ equal $\mathcal{B}(S_p)$,  with $I_F' \subsetneq I_F''$, the LF sets obey a strict hierarchy $\mathcal{LF}(S_p,I_F'')\subsetneq \mathcal{LF}(S_p,I_F')$.
    }
\end{figure} 

The main results up to Theorem \ref{LFinclusionTheorem} in terms of their geometric implications are illustrated in Fig.~\ref{fig:Firstresultsfig}. Fig.~\ref{fig:decisiontreefig} illustrates what kind of properties of a set $\mathcal{LF}(S)$ can be extracted relative to the sets $\mathcal{B}(S)$ and $\mathcal{NS}(S)$ based on the specification of the sets in $S$.

\begin{figure}
    \centering
    \includegraphics[width=\columnwidth]{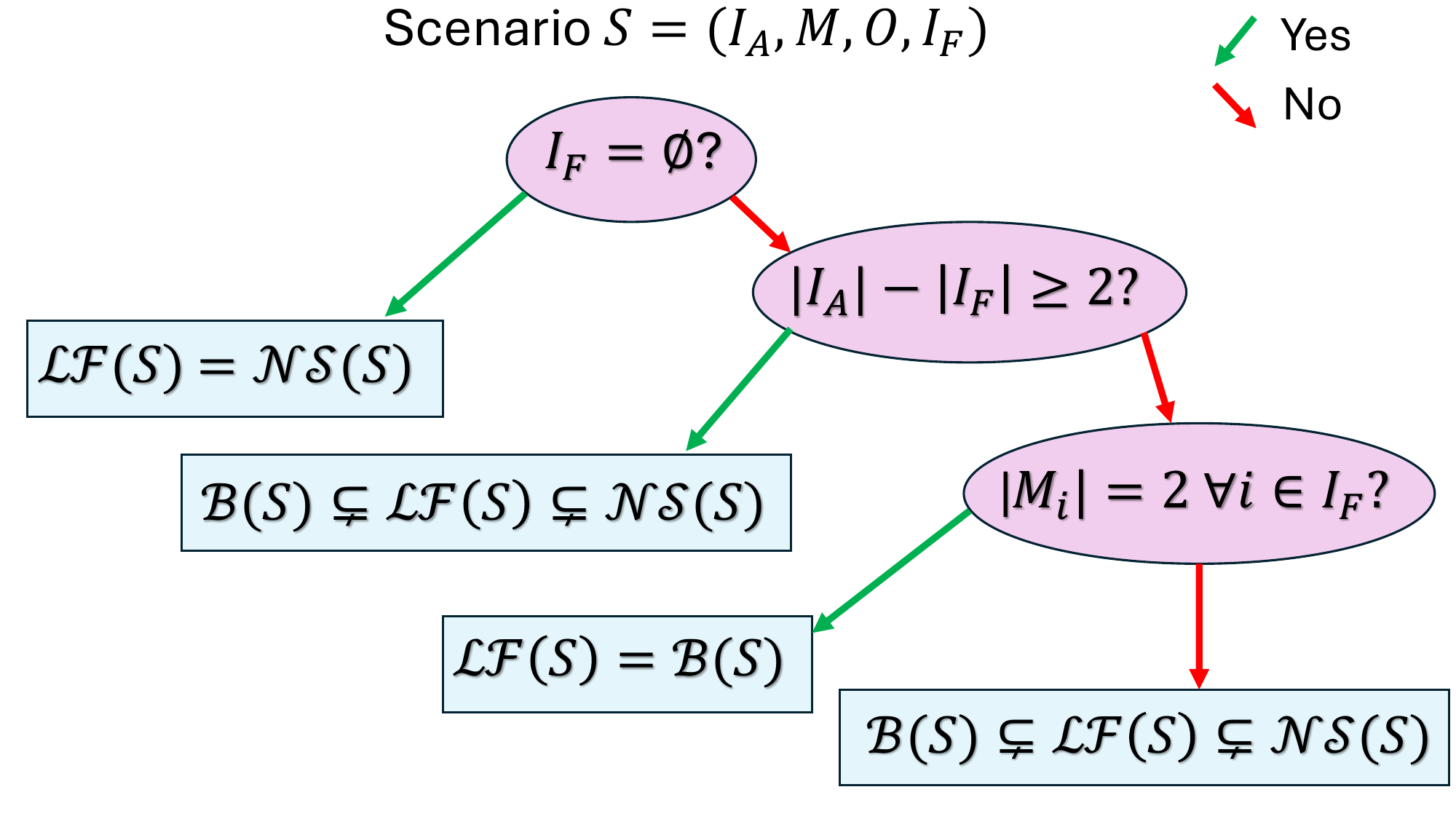}
    \caption{\label{fig:decisiontreefig} A decision tree illustrating what kind of information on the features of the set $\mathcal{LF}(S)$ can be extracted based on the parameters that define a canonical LF scenario $S=(I_A, M, O, I_F )$.
    }
\end{figure} 

Similar arguments can be also used to establish results regarding how LF polytopes relate to each other in scenarios where the sets $I_F$ and  $I_F'$ are at least partly disjoint. 

\begin{theorem} 
    \label{DisjointRelationTHRM}Let $S=(S_p,I_F)$ and $ S'=(S_p,I_F')$ be two otherwise identical scenarios but with $I_F \neq I_F'$. If $ I_F \not \subset I_F'$ and $I_F' \not \subset I_F$, then exactly one of the three following conditions holds true:
    \begin{enumerate}[label=\roman*), ref=\ref{DisjointRelationTHRM}.\roman* ]
        \item $|I_F| = |I_F'| = N-1$, $|M_i| = 2$ $ \forall i$ and  $\mathcal{LF}(S) = \mathcal{LF}(S') = \mathcal{B}(S_p)$ \label{disjointRelationResult1}

        \item  $ \mathcal{B}(S_p) =\mathcal{LF}(S) \subsetneq  \mathcal{LF}(S')$ or $ \mathcal{B}(S_p) = 
 \mathcal{LF}(S') \subsetneq \mathcal{LF}(S)$ \label{DisjointRelationResult2}

        \item  $\mathcal{LF}(S) \not \subset \mathcal{LF}(S')$ and $\mathcal{LF}(S') \not \subset  \mathcal{LF}(S)$. \label{DisjointRelationResult3}
    \end{enumerate} 
\end{theorem}
\begin{proof}
\noindent

\begin{enumerate}[label=\roman*)]
         \item   $\mathcal{LF}(S) = \mathcal{LF}(S') = \mathcal{B}(S_p)$ clearly holds if conditions on the measurements that lead to Theorem \ref{N-input-forAlice-two-inputOthersThreom1Charliemissing} hold for both scenarios. It is sufficient to note that for this to be possible, by appeal to the assumed conditions on $I_F$ and $I_F'$,  it must be that  $|I_F| = |I_F'| = N-1$, $|M_i| = 2$ $ \forall i$.
         \item The second condition \ref{DisjointRelationResult2} clearly holds by appeal Theorem \ref{NoLHVtheorem}  if the premiss of Theorem \ref{N-input-forAlice-two-inputOthersThreom1Charliemissing} holds for one scenario but not the other.  For example if  $|I_F| = |I_F'| = N-1$, $|M_i| \geq 3$ for exactly one $i$, $|M_j|=2$ for $\forall j\neq i$ and additionally $i \notin I_F, i \in I_F'$ then $\mathcal{LF}(S) = \mathcal{B}(S_p) \subsetneq \mathcal{LF}(S')$, by Theorems \ref{N-input-forAlice-two-inputOthersThreom1Charliemissing} and \ref{NoLHVtheorem}. The   other case is true if $i\notin I_F', i \in I_F$. In fact, the same argument can easily be used to establish more generally that if either $|I_F'| < |I_F| =N-1$ or $|I_F| < |I_F'| = N-1$ and the conditions of Theorem $\ref{N-input-forAlice-two-inputOthersThreom1Charliemissing}$ hold for the scenario with $N-1$ elements then either $\mathcal{B}(S_p) = \mathcal{LF}(S) \subsetneq \mathcal{LF}(S')$ or $\mathcal{B}(S_p) = \mathcal{LF(S') \subsetneq \mathcal{LF}(S)}$ respectively.
         
         \item Lastly, we will show with similar arguments as in proof of Theorems \ref{NoLHVtheorem} and \ref{LFinclusionTheorem}, that in all other scenarios the last condition \ref{DisjointRelationResult3}, that $\mathcal{LF}(S) \not \subset \mathcal{LF}(S')$ and $\mathcal{LF}(S') \not \subset \mathcal{LF}(S)$, holds.  The remaining situations can, for the purposes of this proof,  be grouped to a single class; an analogous argument can be constructed in any such case. Here the scenarios are such that there exists at least one $i$ such that $|M_i| \geq 3$  while simultaneously $i \in I_F, I_F'$, so that neither polytope admits an LHV model.

         Now by hypothesis there always exist at least one pair of indices $s,t \in I_A$ such that $s\in I_F, s \notin I_F'$ and $t\in I_F', t \notin I_F$. 
        Therefore a behaviour for which  parties  $(t, i)$ share PR-correlations for their inputs $ x_t \in \{1,2 \}, x_i \in \{2,3\}$, say, is consistent with the LF model in scenario with $I_F$, while in the scenario $I_F'$ by Theorem \ref{N-input-forAlice-two-inputOthersThreom1Charliemissing} this two-party sub-behaviour would admit an LHV model, hence $\mathcal{LF}(S) \not \subset \mathcal{LF}(S') $. An analogous example where the parties $(s,i)$ with inputs $x_s \in \{1,2 \}, x_i \in \{2,3 \}$ share a PR box in scenario $I_F'$ can be used to demonstrate the converse, that also $\mathcal{LF}(S') \not \subset \mathcal{LF}(S) $.
     \end{enumerate}
\end{proof}

Theorem \ref{DisjointRelationTHRM} completes the picture with respect to classifying the sets that bound Local Friendliness behaviours in standard scenarios. It is worth noting that due to the asymmetry of the conditions implied by Local Friendliness-- which are exemplified in the form of Theorem \ref{N-input-forAlice-two-inputOthersThreom1Charliemissing}-- some attention in grouping Local Friendliness polytopes by their symmetries needs to be exhibited. Theorem \ref{DisjointRelationTHRM} shows that in general when the sets of friends are disjoint, one may expect the polytopes to be disjoint as well. While it is always possible in such cases to relabel the sets of friends so that one is included in the other, one should not expect, that such a relabeling correctly maps the inequalities bounding those polytopes in respective scenarios -- instead, one should consider permutations of superobserver-observer pairs, or permutations of friends in subgroups in which the numbers of measurements are equal.

\section{Discussion \label{DiscussionSection}}
In this section we present a few examples  that we believe to be of particular interest and  discuss some further implications of the results in Section \ref{ResultsSection}.

\begin{example} \label{MinimalLFexample}
(LHV model for LF behaviour in the minimal nontrivial Local Friendliness scenario)\\

Consider a canonical scenario with $I_A = \{1,2 \}$, $I_F = \{1 \}, M_1 = M_2 = \{1,2 \}$ and $O_{x_i} = \{-1, 1 \}$ $\forall x_i$. This is the smallest nontrivial Local Friendliness scenario. By Theorem \ref{N-input-forAlice-two-inputOthersThreom1Charliemissing} all behaviour compatible with the principle of Local Friendliness are LHV modelable, so testing the CHSH-inequalities which bound the CHSH-Bell polytope is sufficient. We demonstrate the usefulness of our construction by explicitly building the LHV model in this case. 

We follow the convention of referring to superobserver 1  as Alice, superobserver 2 as Bob and Alices friend as Charlie. The inputs and outputs of Alice and Bob respectively are denoted by $x, y$ and $a, b$ while the output of Charlie is denoted by $c$. 

The assumption of Local Friendliness means that the behaviour $\wp(a,b|x,y)$ decomposes as
    \begin{align}\label{CHSH-NinputBob}
        \wp(a,b|x,y) = \begin{cases}
        \sum_c \delta_{a,c} P^{\mathcal{LA}}(b|y, c) P(c) & \textrm{ when } x=1\\
        \sum_c P^{\mathcal{LA}}(a,b|x,y,c)P(c) & \textrm{ otherwise} \textrm.
    \end{cases}
    \end{align}
the distribution $P(b|y,c)$ has to satisfy 

\begin{align}
    P^{\mathcal{LA}}(b|y,c) = \sum_a P^{\mathcal{LA}}(a,b | x=2, y, c) 
\end{align}
for any $y \in \{1,2 \}$. This means that providing a separable model for the 2-distribution  set $ \{P^{\mathcal{LA}}(a,b| x=2, y, c) \}$ is sufficient, since every other context is already separable by eq. (\ref{CHSH-NinputBob}). The construction is the same as in Theorem \ref{N-input-forAlice-two-inputOthersThreom1Charliemissing}, and we present it without further discussion. Namely one may define
\begin{align}
\begin{split}
    & P(\alpha_{x = 2},\alpha_{y=1},\alpha_{y=2}|c) := \\
    & \dfrac{P^{\mathcal{LA}}(a,b_1 | x=2, y=1, c)  \cdot P^{\mathcal{LA}}(a, b_2 | x=2, y=2, c)}{P^{\mathcal{LA}}(a|x=2, c)},
\end{split}
\end{align}
which provides a joint distribution over the outcomes of Bob in every context, and the outcomes of Alice for the measurement labelled by 2. This means that one may write a joint distribution $P(\lambda)$ over the outcomes of the whole experiment by
\begin{align}
    P(\lambda) =  \sum_{c} P(\alpha_{x = 2},\alpha_{y=1},\alpha_{y=2}|c) \cdot \delta_{\alpha_{x=1}, c} P(c).
\end{align}
This provides a deterministic LHV model for the behaviour via 

\begin{equation}
    \wp(a,b|x,y) = \sum_\lambda \delta_{a,\alpha_x(\lambda)} \cdot \delta_{b, \alpha_y(\lambda)} \cdot P(\lambda),
\end{equation} where $\alpha_{x/y}(\lambda) \in O_{x/y}$. The number of outcomes did not need enter the argument, and by the construction of Theorem \ref{N-input-forAlice-two-inputOthersThreom1Charliemissing} the same proof can be generalized to any number of settings for Bob, but not for Alice. 
\end{example}

Example \ref{MinimalLFexample} is interesting because it portrays the smallest nontrivial  standard LF scenario and shows that any behaviour compatible with the assumption of Local Friendliness  can always be given an LHV model in such scenarios. As mentioned before, this result, and the fact that it generalizes to arbitrary number of settings for Bob, may  already be established as a consequence of the results in ref.~\cite{Woodhead2014} presented in the context of partially deterministic polytopes. Our technique differs from that approach however, as our construction directly shows how an LHV model may be constructed, given behaviour compatible with Local Friendliness, or equivalently a form of partial determinism. 

In Section \ref{ResultsSection} we pointed out that the set $\mathcal{LF}(S)$ always contains behaviour that are not in the quantum set $\mathcal{Q}(S)$, unless $\mathcal{LF}(S)= \mathcal{B}(S)$. Since the minimal scenario  of example \ref{MinimalLFexample}, or by  Corollary \ref{N-input-forAlice-two-inputOthersThreom} the variant with an additional friend,  can always be considered as a subscenario of a canonical LF scenario when $I_F \neq \emptyset$ in the sense of lemmas \ref{ReducedPartiesLHV} and \ref{ReducedMeasurementsLHV},  it follows that in such scenario at least some bipartite sub-behaviour of any $\wp(\Vec{a}|\Vec{x}) \in \mathcal{LF}(S)$ must be LHV-modelable. On the other hand it is easy to construct quantum counter-examples to such sub-behaviour, by allowing those two parties to share correlations that violate one of the CHSH-Bell inequalities for the restricted settings.  For ease of reference, we gather these results into the following corollary.

\corollary{\label{QuantumLFcorollary}Let $S$ be a standard LF scenario. Then one of the three conditions for the relationship between the quantum set $\mathcal{Q}(S)$  and Local Friendliness polytope $\mathcal{LF}(S)$ holds:
\begin{enumerate}[label=\roman*)]
    \item  $\mathcal{Q}(S) \subsetneq \mathcal{LF}(S) = \mathcal{NS}(S)$
    \item  $\mathcal{B}(S) = \mathcal{LF}(S) \subsetneq \mathcal{Q}(S)$
    
    \item $\mathcal{LF}(S) \not\subset \mathcal{Q}(S)$ and $\mathcal{Q}(S) \not\subset \mathcal{LF}(S)$.
\end{enumerate}
}
\begin{proof}
    Omitted. See the discussion following the proof of Theorem \ref{NoLHVtheorem} and the second paragraph following example \ref{MinimalLFexample}.
\end{proof}

\normalfont

This result is illustated in Fig.~\ref{fig:QuantumLFresultsfig}.

\begin{figure}
    \centering
    \includegraphics[width=\columnwidth]{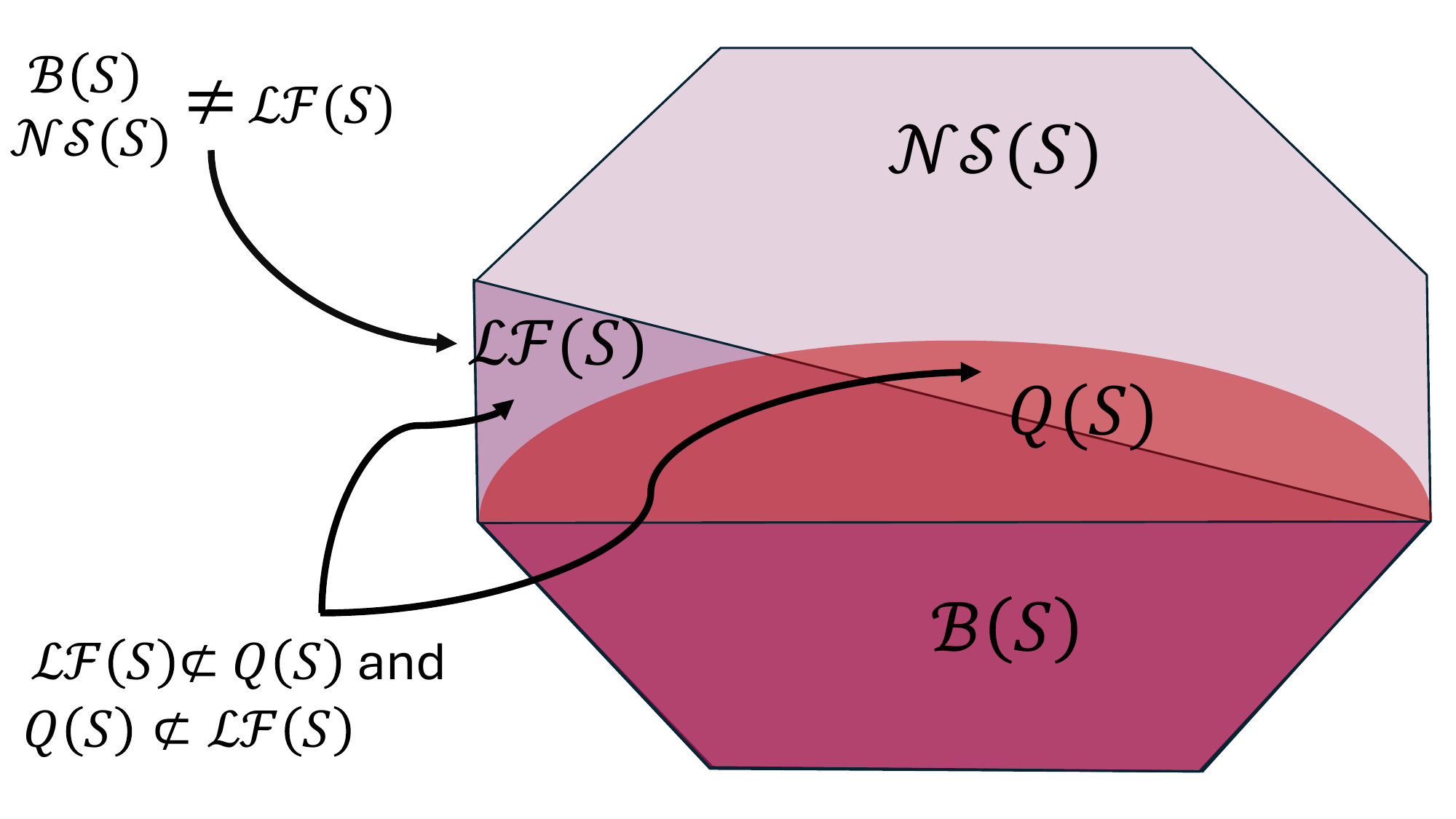}
    \caption{\label{fig:QuantumLFresultsfig}
A depiction of the relationship between the quantum set $Q(S)$ and $\mathcal{LF}(S)$ as established in Corollary \ref{QuantumLFcorollary}. Except in the extreme cases where $\mathcal{LF}(S)$ equals the no-signaling polytope $\mathcal{NS}(S)$  (when $Q(S) \subsetneq \mathcal{LF}(S)$) and where $\mathcal{LF}(S)$ equals the Bell polytope $\mathcal{B}(S)$ (when $\mathcal{LF}(S) \subsetneq  Q(S)$), in general, the quantum set is neither contained or fully contains $\mathcal{LF}(S)$. 
    }
\end{figure} 

The next example demonstrates the implications of Local Friendliness in a multipartite $N\geq 3$ scenario. 

\begin{example} \label{Threeparty-FriendLFexample}
    (LHV model for LF behaviour in 3-partite scenarios with $I_A = I_F = \{1,2,3 \}$, $M_i = \{1,2 \}$ $\forall i$, $|O_{x_i=s}| = 2 $ $\forall i, s$.)\\

    This case is a special case of that in Corollary \ref{N-input-forAlice-two-inputOthersThreom} and we will be brief here. In this case the joint distribution  ${P}(\vec{\alpha}|\vec{c})$ of Eq.~\eqref{massiveJoint2} collapses to a single distribution, namely
    \begin{align}
        {P}(\vec{\alpha}|\vec{c}) = P^{\mathcal{LA}}(a_{x_1=2},a_{x_2=2},a_{x_3=2}|\Vec{x}=222, \Vec{c}).
    \end{align}
Therefore  one may construct a joint distribution over all outcomes of the experiment and build an LHV model via setting $ P(\lambda) = \sum_{\Vec{c}} {P}(\vec{\alpha}|\vec{c}) \cdot  \prod_{i} \delta_{a_{x_i=1}, c_i}\cdot P(\vec{c})$ so the behaviour admits an LHV model of the form 
\begin{align}
    \wp(\vec{a}|\Vec{x}) = \sum_{\lambda} \prod_{i} \delta_{a_i, \alpha_{x_i}(\lambda)} \cdot P (\lambda),
\end{align}
with $\alpha_{x_i}: \Lambda \rightarrow O_{x_i}.$
\end{example}

As stated before, the consequences of Local Friendliness in multipartite scenarios have previously been largely unexplored in the literature. Our results identify all the scenarios in which the multipartite LF polytopes match the Bell polytopes and where they do not.  Example \ref{Threeparty-FriendLFexample} in particular may also be looked at in light of a recent work \cite{Ding2023}, which investigates violations of a tripartite correlation inequality as demonstrations of Brukner's no-go theorem \cite{Brukner2018} and the Local Friendliness no-go theorem \cite{Bong2020} from the perspective of tests on a quantum computer in such a scenario. They claim  \cite{Ding2023} that the same kind of tripartite correlation inequality has to hold for behaviour consistent with assumptions leading to both, respectively. We agree with this statement, based on the fact that the LF polytope equals the Bell polytope in this set up. We do however point out to discussions elsewhere \cite{Bong2020, Healey2018}, where it has been established that Brukner's no-go result \cite{Brukner2018} implicitly uses the assumption of noncontextuality in the sense of Kochen and Specker \cite{Kochen1967} in its derivation. Therefore behaviour compatible with Brukner's assumptions \cite{Brukner2018} will always obey the Bell inequalities, which by Theorem \ref{NoLHVtheorem}  is not the case for behaviour compatible with Local Friendliness in general. 

Theorem \ref{N-input-forAlice-two-inputOthersThreom1Charliemissing} implies that the experiment described in example \ref{Threeparty-FriendLFexample} could also be further simplified by removing one of the friends to match the set up portrayed in Figs.~\ref{fig:ThreepartyScenario} and \ref{fig:spacetimeThreeparty}, while still leading to the same set of inequalities. We believe these type of implications of our work to be of interest from an experimental perspective, as on one hand, there is motivation to keep the experimental designs as simple as possible, while on the other one would like to make sure no computational resources are unnecessarily dedicated for the purpose of seeking inequalities demonstrating a result of interest. 

While our main motivation for this work was exploring the general bounds imposed by Local Friendliness, we stress that our main results are purely  mathematical in their nature. As mentioned before, the bipartite Local Friendliness polytopes  had  previously been investigated from the perspective of device independent randomess certifiation in ref.~\cite{Woodhead2014}, where they may be identified as examples of objects termed  partially deterministic polytopes. One may thus conjecture, that the results of this work regarding bounds on the set of behaviour compatible with Local Friendliness could be of use in further contexts of device or semi -device independent information processing.

\section{Conclusion\label{ConclusionSection}}

 In this work, we have explored the mathematical properties of Local Friendliness behaviours in multipartite generalizations of the scenarios considered in ref.~\cite{Bong2020}. We have demonstrated that, analogously to the set of behaviours that admit a Local Hidden Variable model~\cite{Brunner2014}, the set of Local Friendliness behaviours can always be identified as a convex polytope.

 Among our results, we have completely characterised the scenarios in which the set of behaviours compatible with Local Friendliness assumptions admits an LHV model. Furthermore, our proofs are constructive, in the sense that we have shown how to build an LHV model if an LF model is given. We have completely characterised the inclusion relations between the set of Local Friendliness correlations of any scenario and the sets of LHV, quantum and no-signalling correlations. We have also completely characterised the inclusion relations between the LF polytopes for different scenarios. As corollaries we have demonstrated, for example, that in any scenario that contains at least one friend the Local Friendliness polytope is strictly contained in the no-signalling polytope, and that there always exists quantum correlations outside of the Local Friendliness set in any such nontrivial scenario. 
  
These results increase our understanding of the structure of LF correlations and their relationship with other well-known sets of correlations. They may also find applications in guiding experimental designs; for example, our results on the hierarchies of the sets of LF behaviours indicate when and if a given experiment may be simplified by removing some friends without losing any ability to demonstrate LF violations. Given that the complexity of vertex or facet enumeration problems for convex polytopes grows very quickly with the number of parameters, our analytical results may also help simplify such computational problems by pruning their search trees.

Some instances of Local Friendliness polytopes in canonical scenarios have also been recognised to be relevant in seemingly unrelated questions of device-independent information processing, under the name of ``partially deterministic polytopes'' \cite{Woodhead2014}. Since our main results concern the mathematical structure of LF polytopes, independently of the motivation from the foundational questions surrounding the Local Friendliness no-go result, we believe that the conclusions of this work could find use in those applications as well.

\section*{ACKNOWLEDGEMENTS}
MH would like to acknowledge useful discussions with Anibal Utreras-Alarcón with regard to scope of previously known results in bipartite cases. This work was supported by the Australian Research Council Future Fellowship FT180100317 and by the ARC Centre of Excellence for Quantum Computation and Communication Technology (CQC2T) project number  CE170100012. When finishing this manuscript we became aware of a concurrent work in ref. \cite{Walleghem2024} which also studied the relationship between the LF and Bell polytopes.

\bibliography{LFpolytopes.bib}
\end{document}